\documentclass[11pt]{article}
\usepackage{authblk}

\input{Qcircuit}
\usepackage{color}
\usepackage{amsmath,amsthm,amsbsy,amssymb,dsfont}
\usepackage{mathtools}
\usepackage[colorlinks=true,urlcolor=webblue,linkcolor=webgreen,filecolor=webblue,citecolor=webgreen,pdfpagemode=UseOutlines,pdfstartview=FitH,pdfpagelayout=OneColumn,bookmarks=true]{hyperref}
\usepackage{fullpage}
\usepackage{doi}
\usepackage[numbers]{natbib}
\usepackage{microtype}
\usepackage{xspace}
\usepackage{tikz}
\usepackage{bigints}
\usepackage{mathptmx}
\usepackage{eucal}

\hypersetup{pdfauthor={Gorjan Alagic}}

\definecolor{webgreen}{rgb}{0,.5,0}
\definecolor{webblue}{rgb}{0,0,.5}

\DeclareMathOperator{\tr}{Tr}

\numberwithin{equation}{section}

\newtheorem{theorem}{Theorem}
\newtheorem{prop}{Proposition}
\newtheorem{proposition}{Proposition}
\newtheorem{lemma}[theorem]{Lemma}
\newtheorem{conjecture}{Conjecture}

\newtheorem{definition}{Definition}

\newtheorem{problem}{Problem}

\newcommand{\one}{\mathds 1}

\DeclareMathOperator{\End}{End}

\newcommand{\opn}{\operatorname}

\newcommand{\C}{\mathbb{C}}

\newcommand{\expref}[2]{\texorpdfstring{\hyperref[#2]{#1~\ref{#2}}}{#1~\ref{#2}}}
\newcommand{\expreft}[2]{\texorpdfstring{\hyperref[#2]{#1}}{#1}}

\newcommand{\revise}[1]{}

\newcommand{\algo}{\mathcal}
\newcommand{\negl}{\opn{negl}}
\newcommand{\poly}{\opn{poly}}
\newcommand{\KeyGen}{\ensuremath{\mathsf{KeyGen}}\xspace}
\newcommand{\Enc}{\ensuremath{\mathsf{Enc}}\xspace}
\newcommand{\Dec}{\ensuremath{\mathsf{Dec}}\xspace}
\newcommand{\Homorcl}{\ensuremath{\mathsf{Hom}}\xspace}
\newcommand{\Mint}{\ensuremath{\mathsf{Mint}}\xspace}
\newcommand{\Verify}{\ensuremath{\mathsf{Verify}}\xspace}
\newcommand{\inrand}{\in_R} 
\newcommand{\prob}{\opn{Pr}}
\newcommand{\states}{\mathfrak D}
\newcommand\supp{\textbf{supp}}
\newcommand\Eval{\ensuremath{\mathsf{Eval}}\xspace}

\title{On quantum obfuscation}

\author{
Gorjan Alagic \thanks{Department of Mathematical Sciences, University of Copenhagen. \textsf{galagic@gmail.com}}
\qquad
Bill Fefferman \thanks {Joint Center for Quantum Information and Computer Science (QuICS), University of Maryland. \textsf{wjf@umd.edu}}
}

\begin{document}

\maketitle

\abstract{Encryption of data is fundamental to secure communication in the modern world. Beyond encryption of data lies \emph{obfuscation}, i.e., encryption of functionality. It is well-known that the most powerful means of obfuscating classical programs, so-called ``black-box obfuscation,'' is provably impossible~\cite{BGIRSVY01}. For years since, obfuscation was believed to always be either impossible or useless, depending on the particulars of its formal definition. However, several recent results have yielded candidate schemes that satisfy a definition weaker than black-box, and yet still have numerous applications.

In this work, we initialize the rigorous study of obfuscating programs \emph{via quantum-mechanical means.} We define notions of quantum obfuscation which encompass several natural variants. The input to the obfuscator can describe classical or quantum functionality, and the output can be a circuit description or a quantum state. The obfuscator can also satisfy one of a number of obfuscation conditions: black-box, information-theoretic black-box, indistinguishability, and best-possible; the last two conditions come in three variants: perfect, statistical, and computational. We discuss a number of applications, including CPA-secure quantum encryption, quantum fully-homomorphic encryption, and public-key quantum money. We then prove several impossibility results, extending a number of foundational papers on classical obfuscation to the quantum setting. We prove that quantum black-box obfuscation is impossible in a setting where adversaries can possess more than one output of the obfuscator (possibly even on the same input.) In particular, generic transformation of quantum circuits into black-box-obfuscated quantum circuits is impossible. We also show that statistical indistinguishability obfuscation is impossible, up to an unlikely complexity-theoretic collapse. Our proofs involve a new tool: chosen-ciphertext-secure encryption of quantum data, which was recently shown to be possible provided that quantum-secure one-way functions exist~\cite{ABFGSS16}. 

We emphasize that our results leave open one intriguing possibility: black-box obfuscation of classical or quantum circuits into a single, uncloneable  quantum state. This indicates that, in spite of our results, quantum obfuscation may be significantly more powerful than its classical counterpart.

\newpage
\tableofcontents
\newpage

\section{Introduction}\label{sec:intro}

The ability to encrypt data is central to modern communications. Basic methods for performing this task with privately-exchanged encryption keys have been known for hundreds of years. More advanced, public-key encryption methods were developed much more recently, beginning with the work of Merkle~\cite{Merkle78} and Diffie and Hellman~\cite{DH76} in the 1970s. These public-key methods have found widespread practical application in virtually all Internet communications. More advanced theoretical methods for encrypting data, such as fully-homomorphic encryption, have only been discovered recently~\cite{Gentry09}, but show great promise for practical application.

Arguably the most powerful encryption ability is \emph{obfuscation}; this is the ability to \emph{encrypt functionality}. Obfuscation implies (with some caveats) the ability to perform almost any cryptographic task imaginable, including public-key and fully-homomorphic encryption. Unlike in the case of data encryption, our theoretical understanding of obfuscation is still fairly limited. 

To understand obfuscation, it is useful to think about an obvious application: protecting intellectual property in software. In this setting, a software developer wishes to distribute their software to end users. However, the code contains a number of trade secrets which the developer does not want to become public. In order to maintain these secrets, the publisher passes the software through an \emph{obfuscation algorithm} (or obfuscator) prior to publishing. In this application, the obfuscator must be an efficient algorithm that satisfies three core properties: 
\begin{enumerate}\label{def:obf-informal}
\item \emph{functional equivalence:} the input/output functionality of the input program does not change;
\item \emph{polynomial slowdown:} if the input program is efficient, then the output program is efficient;
\item \emph{obfuscation:} the code of the output program is ``hard to understand.''
\end{enumerate}
The last condition can be formulated rigorously in a number of ways. One possibility is the so-called ``virtual black-box'' condition, which says that the obfuscated program is no more useful than an impenetrable box which simply accepts inputs and produces outputs. While this condition appears to be too strong in the classical world, there are other formulations (with varying levels of strength and usefulness) which may be achievable.

The study of encrypting classical data and classical programs is significantly complicated by the advent of \emph{quantum computation.} One well-known consequence of the presence of quantum computers is that certain data encryption schemes, such as those based on the hardness of factoring or the discrete logarithm, are no longer secure. It is conceivable that certain classical obfuscation schemes are also not secure against quantum adversaries. On the other hand, quantum mechanics also appears to enable certain cryptographic tasks (such as information-theoretically secure key exchange) that are impossible classically. It is thus natural to ask what quantum computation means for obfuscation of programs. In particular, we would like to answer the following questions:

\begin{itemize}
\item what are some natural formulations of quantum-mechanical program obfuscation?
\item is it possible to quantumly obfuscate classical and/or quantum programs?
\item which of the classical results about obfuscation carry over to the quantum setting?
\item are there applications of quantum obfuscation that are impossible classically?
\end{itemize}

We remark that, in order to address the above questions, we must also properly address the question of encrypting quantum data---a strictly simpler task than encrypting functionality. While information-theoretic encryption of quantum data has been considered before, in this setting we are interested in encryption of quantum data \emph{with computational assumptions}\footnote{Note that information-theoretic obfuscation is impossible if the adversary can execute the obfuscated program on all possible inputs; indeed, a computationally unbounded adversary can use this ability to learn everything there is to know about the program.}. This latter subject has not yet received significant attention in literature.

Before continuing, we draw attention to the distinction between obfuscating \emph{programs} and obfuscating \emph{circuits}. While these two forms of obfuscation are closely related, there are some important technical differences. In this work, as in most theoretical works on obfuscation, we will focus on obfuscation of circuits. We view the circuit model as more convenient; it also tends to be preferred in the theoretical literature on both cryptography and quantum computation.

\subsection{Background}

We now briefly review the current state of affairs in research on obfuscation and quantum encryption with computational security assumptions. The classical case has been studied significantly. On the other hand, quantum obfuscation has received little to no attention, outside the proposal in~\cite{AJJ14}. Quantum encryption is an active current area of research.

\subsubsection{Classical obfuscation} 

Ad-hoc obfuscation of software has been a fairly common practice for some time. In fact, simply compiling a program can be viewed as a form of obfuscation.  The earliest mention of obfuscation in the modern study of theoretical cryptography appears to be in the famous paper of Diffie and Hellman~\cite{DH76}. There, it was suggested that public-key cryptosystems might be constructible via obfuscation of private-key schemes; this was viewed as a reasonable possibility because writing code in an obfuscated manner seems relatively easy in practice. 

The first major result in classical obfuscation was the 2001 proof by Barak et al. that virtual black-box obfuscation is impossible~\cite{BGIRSVY01, BGIRSVY12}. Their definition is based on the \emph{simulation paradigm}. More precisely, the obfuscation condition (i.e., the third condition in the \expreft{previous section}{def:obf-informal}) states that any efficient algorithm with access to an obfuscated circuit should be simulable by another efficient algorithm with only oracle (i.e., black-box) access to the original functionality. This definition is very natural in the setting of the aforementioned ``software intellectual property protection'' application: the end user can only learn that which is learnable by simply running the program. Barak et al. proved that  there exist circuit families which are unobfuscatable under this definition. They also showed that some of the most sought-after applications of black-box obfuscation are impossible. For instance, they showed that private-key encryption schemes cannot be transformed to public-key ones by obfuscating the encryption circuits in a generic manner.

The years following the Barak et al. result saw some limited progress in theoretical obfuscation. It was proved possible for some limited forms of functionality~\cite{CD08, Wee05}, and some additional limits were placed, e.g., on black-box obfuscation with auxiliary input~\cite{GK05}. An important step in formulating feasible notions of obfuscation was taken by Goldwasser and Rothblum; they defined \emph{indistinguishability obfuscation} and \emph{best-possible obfuscation}~\cite{GR07}. Both of these definitions alter the obfuscation condition, while leaving the functional-equivalence and polynomial-slowdown conditions unchanged. Under indistinguishability, it is required that the obfuscator maps functionally-equivalent circuits to indistinguishable distributions. Under best-possible, the obfuscator maps any circuit to a circuit from which the end user can ``learn the least.'' Both definitions have a perfect, statistical, and computational variant. Goldwasser and Rothblum proved that the two definitions are equivalent, and that the perfect and statistical versions are impossible (unless the PH collapses)~\cite{GR07}. This left one possibility: computational indistinguishability obfuscation. It was widely believed that computational indistinguishability was too weak of a condition to provide any interesting applications.

In 2013, in a breakthrough result, Garg et al. proposed a convincing candidate for computational indistinguishability obfuscation~\cite{GGHRSW13}. They proposed an obfuscation scheme for NC1 circuits, based on the presumed hardness of a problem in multilinear maps; they also showed how to use fully-homomorphic encryption (with NC1 decryption circuits) to ``bootstrap'' their NC1 scheme to obfuscation for all circuits. Around the same time, another breakthrough by Sahai and Waters showed how to use a computational indistinguishability obfuscator to achieve a wide-range of applications, via a new ``punctured programs'' technique~\cite{SW14}. These applications include chosen-ciphertext-secure public-key encryption, injective trapdoor functions, and oblivious transfer. Sahai and Waters suggested that the applications were so wide-ranging that indistinguishability obfuscation might become a ``'central hub' for cryptography''~\cite{SW14}. These two breakthroughs were followed by a flurry of new activity in the area, including several new proposals and applications~\cite{BGKPS14, BCCGKPR14, BZ14, BR14, GGHW14, HSW14}.

\subsubsection{Quantum obfuscation} 

Quantum obfuscation is essentially an unexplored topic, and the present work appears to be the first rigorous treatment of the foundational questions. The question of whether quantum obfuscation is possible was posed as one of Scott Aaronson's ``semi-grand challenges'' for quantum computation~\cite{Aar05}. Since so little work on quantum obfuscation has appeared, our brief discussion will also mention some results that we believe are related. 

In~\cite{Aar09}, Aaronson proposed two relevant results. The first was a \emph{complexity-theoretic no-cloning theorem}, stating that cloning an unknown, random state by means of a black-box ``reflection oracle'' requires exponentially many queries. The second theorem stated that an oracle exists relative to which ``software copy-protection'' is possible. Unfortunately, a full version of~\cite{Aar09} with proofs never appeared, although the complexity-theoretic no-cloning theorem was eventually proved in a paper on quantum money~\cite{AC12}. In related work, Mosca and Stebila proposed a black-box quantum money scheme, and suggested the possibility of using a quantum circuit obfuscator in place of the black box~\cite{MS10}.

More recently, Alagic, Jeffery and Jordan proposed obfuscators for both classical (reversible) circuits and quantum circuits, based on ideas from topological quantum computation~\cite{AJJ14}. The proposed obfuscator compiles the circuits into braids using certain high-dimensional representations of the braid group, and then applies an algorithm for putting braids into normal form. Although it is efficient, this algorithm does not satisfy any of the aforementioned obfuscation definitions; instead, it satisfies perfect indistinguishability for a restricted set of circuit equivalences. The usefulness of such an obfuscator is unclear at this time.

\subsubsection{Quantum encryption}

In order to discuss quantum obfuscation of functionality, we require a more basic primitive: quantum encryption of data. As we will see later, information-theoretic encryption is insufficient for our purposes, and we must thus rely on computational assumptions. In particular, we require (i.) reusability of the key, and (ii.) chosen-ciphertext security. Recent results have shown that this is possible to achieve under the assumption that quantum-resistant one-way functions exist~\cite{ABFGSS16}. For readability, we briefly summarize the relevant results from~\cite{ABFGSS16} below, and in more detail in \expref{Section}{sec:encryption}.
\begin{enumerate}
\item \textbf{Quantum encryption schemes.} One can define a notion of symmetric-key encryption scheme for quantum states, with reusable keys; these schemes consist of three quantum algorithms (key generation, encryption, and decryption) which satisfy correctness: under a fixed key, encryption followed by decryption must be equivalent to the identity. Such schemes first appeared in~\cite{BJ15}.
\item \textbf{Chosen-ciphertext security for quantum encryption.} One may also define IND-CCA1 security (or \emph{indistinguishability of ciphertexts under non-adaptive chosen ciphertext attacks}) for these schemes; this formalizes the idea of a ``lunchtime attack,'' where an adversary has complete access to all aspects of the encryption scheme except the key itself, and is tasked with decrypting a challenge ciphertext later (presumably after lunch.) 
\item \textbf{An IND-CCA1-secure construction.} If quantum-secure one-way functions (qOWF) exist, then so do IND-CCA1-secure symmetric-key encryption schemes for quantum states. These qOWFs are deterministic classical functions which are easy to compute, but hard to invert for quantum adversaries. 
\end{enumerate}

\subsection{Summary of results}

In this section, we summarize our results and discussions. These are divided by subject, with quantum encryption covered in \expref{Section}{sec:encryption}, quantum black-box obfuscation in \expref{Section}{sec:black-box}, and quantum indistinguishability obfuscation in \expref{Section}{sec:indistinguishability}.

\subsubsection{Quantum black-box obfuscation}

\paragraph{Definitions.} Our main results concern definitions, applications, and (im)possibility of quantum obfuscation in the virtual black-box setting. We will begin by defining the following.

\begin{enumerate}
\item \textbf{Quantum black-box obfuscator.} This is a polynomial-time quantum algorithm $\algo O$ which accepts quantum circuits $C$ as input, and produces quantum states $\algo O(C)$ as output. It preserves functionality, in the sense that there is a publicly known way to use $\algo O(C)$ and any input state $\ket{\psi}$ to produce the state $C \ket{\psi}$\,. It satisfies a black-box condition, which states that for polynomial-time quantum algorithms, possession of $\algo O(C)$ can be simulated by black-box access to $C$. This definition is a natural analogue of the classical black-box definition given in~\cite{BGIRSVY12}.
\item \textbf{Quantum ``two-circuit'' black-box obfuscator.} This obfuscator is precisely as above, except the obfuscation condition is strengthened to hold over arbitrary \emph{pairs of circuits} $(C_1, C_2)$. For us, this definition will be primarily useful because of its role in establishing certain impossibility results.
\item \textbf{Information-theoretic quantum black-box obfuscator.} This is a modification of the above definition, in which we posit that \emph{any} adversary with access to $\algo O(C)$ can be simulated by a \emph{polynomial-time} quantum simulator with black-box access to $C$. This definition is impossible classically, for obvious reasons: both $\algo O(C)$  and $\algo O$ can be copied and reused an arbitrary number of times, enabling unbounded adversaries to discover everything about $C$.
\end{enumerate}

\paragraph{Impossibility.} We prove three impossibility results, which place several important restrictions on quantum obfuscation. Our impossibility proofs are based on the ideas of Barak et al.~\cite{BGIRSVY12}, with several important quantum adaptations, and a new quantum ingredient: the aforementioned IND-CCA1 quantum encryption.
\begin{enumerate}
\item \textbf{Two-state black-box obfuscation is impossible.} We prove that there exist families of circuit pairs which can reveal a secret if one is in possession of a circuit description for both of them, but not if one only has black-box access. This impossibility persists even if the obfuscation output is a quantum state, as opposed to a circuit description. Unlike the other results, it is also true even if the obfuscated states are \emph{not reusable.}
\item \textbf{If qOWFs exist, then obfuscation with more than one output is impossible.} For this proof, we combine the pairs from the circuit families in the two-circuit impossibility proof in order to build a single unobfuscatable family. The ability to execute obfuscated states from this family \emph{on themselves} is crucial here, and has two requirements: (i.) access to more than one obfuscation, even if the obfuscations are quantum states, and (ii.) secure encryption, which in turn requires the existence of qOWFs. This result applies both to quantum black-box obfuscators (as in the first definition above) and the information-theoretic variant (as in the third definition above.) 
\item \textbf{Classical algorithms for quantum obfuscation are impossible, unconditionally.} This result follows directly from the previous result and Application 1 below. It can be viewed as an extension of the original Barak et al. impossibility result to the case of quantum functionality and quantum adversaries.
\end{enumerate}
We emphasize that, while our techniques are very similar to Barak et al., our results are not a simple consequence of the fact that classical functions (and in particular, the Barak et al. unobfuscatable functions) are special cases of quantum functions. This ``special case'' argument fails for obvious reasons when the output of the obfuscator can be a quantum state. As it turns out, it also fails when the output is a quantum circuit; briefly, the reasons are (i.) it is now permitted to specify (even approximate) quantum circuits for classical functions, (ii.) the black-box simulator now has quantum access, which is significantly more powerful than classical (see, e.g., ~\cite{HR14}) and (iii.) the Barak et al. adversary is insufficient since it can only perform classical gates homomorphically. The last point is where the aforementioned quantum encryption tools become necessary.

\paragraph{Applications.} We then move on to discuss potential applications of quantum black-box obfuscators. We emphasize that (with the exception of the first one), all of these applications are still possible \emph{in some form} in spite of the above impossibility result. We view this as a strong indication that quantum obfuscation should be studied further. While some of the applications are analogues of known classical applications (as outlined in \cite{BGIRSVY12},) the last is special to the quantum setting. We are certain that many other quantum-specific applications are possible, given the combined advantage of obfuscation and no-cloning.

\begin{enumerate}
\item \textbf{Quantum-secure one-way functions.} We show that, if there exists a classical probabilistic algorithm for quantum obfuscation, then quantum-secure one-way functions exist. The above impossibility result rules this out, but the implication is nonetheless interesting; for one, it enables the very proof of the impossibility result itself! The one-way functions are essentially the functions computed by the obfuscator (with fixed randomness) on circuits with a ``hidden output.'' We are unable to extend this application to the setting of efficient quantum algorithms for obfuscation. We leave this as an interesting open problem, and note its connection to developing foundational primitives for quantum encryption.
\item \textbf{IND-CPA-secure private-key quantum encryption.} In this application, the obfuscation algorithm can be quantum; moreover, we do not demand the existence of one-way functions or any other primitive.
\item \textbf{qOWF imply IND-CPA public-key encryption.} This application combines IND-CCA1-secure private-key encryption (which follows from qOWFs) with obfuscation of the encryption circuits. The result is public-key encryption of quantum states without the need for trapdoor permutations (unlike in~\cite{ABFGSS16} and, indirectly, in~\cite{BJ15}.)
\item \textbf{qOWF imply IND-CPA quantum fully homomorphic encryption.} This application combines the previous application, together with obfuscation of a universal decrypt-compute-encrypt circuit. Depending on the properties of the obfuscator, it may also satisfy \emph{compactness} (the requirement that communication between client and server does not scale with the size of the computation.)
\item \textbf{Public-key quantum money}. Using circuit obfuscation to produce public-key quantum money was first proposed by Mosca and Stebila~\cite{MS10}, using a complexity-theoretic no-cloning theorem proposed by Aaronson~\cite{Aar09} and proved by Aaronson and Christiano~\cite{AC12}. We outline the ideas here, and discuss the new limitations placed by our results.
\end{enumerate}

We emphasize that all the above applications except quantum money also work for achieving \emph{classical functionality} from a quantum obfuscator; however, depending on the details of the obfuscator and the application, this may require quantum algorithms for encryption and decryption, or even quantum ciphertexts.

\subsubsection{Quantum indistinguishability obfuscation}

Lastly, we consider an alternative formulation of obfuscation, motivated by the classical definitions of indistinguishability obfuscation and best-possible obfuscation, as set down by Barak et al.~\cite{BGIRSVY12} and Goldwasser and Rothblum~\cite{GR07}. We establish quantum analogues of the central results in those classical papers. In this setting, rather than comparing the obfuscation of the circuit to that of a black-box, we compare it to the obfuscations of other, functionally-equivalent circuits. Starting with the new definitions, our results are as follows.

\begin{enumerate}
\item \textbf{Quantum indistinguishability obfuscator.} Just as in the black-box definition, this is a polynomial-time quantum algorithm $\algo O$ which accepts quantum circuits $C$ as input, and produces ``functionally-equivalent'' quantum states $\algo O(C)$ as output. The obfuscation condition now states that functionally equivalent circuits are mapped to \emph{indistinguishable} states. Based on the kind of indistinguishability deployed in the definition, there are three variants of an indistinguishability obfuscator: perfect, statistical, and computational.
\item \textbf{Quantum best-possible obfuscator.} This is an algorithm precisely as above, except for the obfuscation condition: it now states that $\algo O(C)$ is the state that ``leaks least,'' among all states which are ``functionally-equivalent'' to $C$. There are again three variants: perfect, statistical, and computational.
\item \textbf{Equivalence of definitions.} We prove that each of the three variants of quantum indistinguishability obfuscation is equivalent to the analogous variant of quantum best-possible obfuscation, so long as the obfuscator is efficient.
\item \textbf{Impossibility of perfect and statistical indistinguishability obfuscation.} We end with a quantum version of the main result of~\cite{GR07}: a proof that perfect and statistical quantum indistinguishability obfuscation is impossible, unless coQMA is contained in QSZK. We remark that an analogous containment in the classical setting (i.e., coMA $\subseteq$ SZK) would imply a collapse of the polynomial-time hierarchy to the second level. Moreover, for the case of obfuscating arbitrary quantum computations (i.e., completely positive, trace-preserving maps), we obtain that statistical quantum indistinguishability obfuscation would imply that PSPACE is contained in QSZK. One consequence of these results is that extending the obfuscator proposed in~\cite{AJJ14} to full indistinguishability is impossible, barring highly unlikely collapses of complexity classes.
\item \textbf{Application: witness encryption for QMA.} Motivated by an analogue discussed in~\cite{GGSW13, GGHRSW13}, we show that a quantum indistinguishability obfuscator enables witness encryption for QMA. A witness encryption scheme for a language $L$ in QMA encrypts plaintexts $x$ using a particular instance $l$. The security condition states that, if $l \in L$, then a valid witness $w$ for $l \in L$ allows decryption; on the other hand, if $l \notin L$, then ciphertexts are indistinguishable. While witness encryption has several applications classically~\cite{GGSW13}, the quantum analogue has not been considered previously.
\end{enumerate}

We remark that, in the classical setting, indistinguishability obfuscation also implies functional encryption~\cite{GGHRSW13} and many more applications through the very successful ``punctured programs'' technique developed by Sahai and Waters~\cite{SW14}. We suspect that these results can also be adapted to the quantum setting, but leave them open for now.

\subsection{Notation and terminology}\label{sec:notation}

In this section, we set down some notation and basic terminology which we will use throughout the rest of the paper.

\subsubsection{Classical}

We will assume that the state space of a classical device can be identified with sets of bitstrings, i.e., $\{0, 1\}^n$ for some positive integer $n$. The notation $x \inrand \{0, 1\}$ will mean that $x$ is an $n$-bit string selected uniformly at random. The set of all bitstrings (of arbitrary length) will be denoted by $\{0, 1\}^*$. Classical functions will then be maps $f : \{0, 1\}^n \rightarrow \{0, 1\}^m$ from one set of bitstrings to another. We will also sometimes consider function families, written $f : \{0, 1\}^* \rightarrow \{0, 1\}^*$; these can be thought of as a function family $\{f_n\}_{n>0}$ indexed by the input size $n$. 

A classical circuit $C$ is a sequence of local boolean gates which, when composed together, implement some (in general irreversible) function $f_C: \{0, 1\}^n \rightarrow \{0, 1\}^m$. The input size of $C$ is $n$, the output size is $m$, and the number of gates is denoted by $|C|$. A probabilistic circuit is also a circuit, but with the input bits divided into two registers: the input register, and the ``coin'' register. A normal execution of a probabilistic circuit involves initializing the coin register with completely random bits, and inserting the input into the input register. We will frequently discuss \textbf{circuit ensembles}; these are infinite families $\{C_n\}_{n > 0}$ of circuits, one for each possible input size, so that the size of circuit $C_n$ is bounded by some fixed polynomial function of $n$.  We say that a circuit ensemble is \textbf{uniform} if there is a deterministic polynomial-time Turing Machine which, on input $1^n$, outputs a classical description of $C_n$.  We will also sometimes make use of \textbf{distributions of circuit ensembles}; these are infinite families $\mathcal C = \{\mathcal C_n\}_{n > 0}$ where each $\mathcal C_n$ is a finite family of circuits of input size $n$, along with a probability distribution $P_{\mathcal C, n}$. For a bitstring $x$, the notation $\mathcal C(x)$ will then denote the probability distribution (on bitstrings) resulting from running a random circuit from the family $\mathcal C_{|x|}$, selected according the distribution $P_{\mathcal C, |x|}$.

A \textbf{deterministic classical algorithm} $\algo A$ is simply a circuit ensemble. Running $\algo A$ on an input bitstring $x$ involves selecting the circuit with the appropriate input size, and executing it with input $x$. If the circuit ensemble is uniform, we will say that $\algo A$ is efficient; more precisely, it is then a classical deterministic polynomial-time algorithm (\textbf{PT} for short.) A \textbf{probabilistic algorithm} $\algo A'$ is an algorithm whose circuits are probabilistic. Running $\algo A'$ on an input bitstring $x$ involves selecting the circuit with the appropriate input size, initializing its coin register with uniformly random bits, and then executing it with input $x$. If the circuits of $\algo A'$ are polynomial-time uniform, we say that $\algo A'$ is an efficient, or classical probabilistic polynomial-time algorithm (\textbf{PPT} for short.) We will frequently use PPTs to model the most general efficient classical algorithms.

\subsubsection{Quantum}

For our purposes, the space of pure states of a quantum device will be identified with a Hilbert space $\mathcal H_n \cong (\C_2)^{\otimes n}$ of a finite number $n$ of qubits. We will identify some fixed orthonormal basis (called the \emph{computational basis}) of $\mathcal H_n$ with the corresponding space $\{0, 1\}^n$ of classical states, so that, e.g, $\ket{x}$ for $x \in \{0, 1\}^n$ denotes a basis element of $\mathcal H_n$. The space of density operators (i.e., general quantum states) of $n$ qubits will be denoted $\states (\mathcal H_n)$; this is the set of positive semidefinite, Hermitian trace-one operators in $\End(\mathcal H_n)$. A state in $\states (\mathcal H_n)$ can be interpreted as a probabilistic mixture $\sum_j p_j \ket{\varphi_j}\bra{\varphi_j}$ of pure states, albeit not in a unique way. We will discuss valid quantum transformations of three types. The first are measurements, which act on a state $\ket{\psi} \in \mathcal H_n$ by projecting some or all of the qubits into the computational basis states $\{\ket{0}, \ket{1}\}$. The second are unitary maps, i.e., linear operators $U: \mathcal H_n \rightarrow \mathcal H_n$ satisfying $U^\dagger U = \one_n$, where $\one_n$ denotes the $n$-qubit identity operator. The third are CPTP maps, i.e., completely positive trace-preserving maps $\Phi : \states (\mathcal H_n) \rightarrow \states (\mathcal H_m)$. CPTP maps are the most general type of evolution, encompassing unitary maps, measurement, and adding or discarding (or tracing out) qubits. For example, a unitary operator $U \in U(\mathcal H_n)$ can be expressed as a CPTP map by writing $\rho \mapsto U\rho U^\dagger$, where $\rho \in \states (\mathcal H_n)$.

A \textbf{quantum circuit} $C$ is a sequence of local unitary gates on a fixed number (say $n$) of qubits; these gates , when composed together, implement some unitary operator $U_C \in U(2^n)$. Definitions of circuit ensembles and distributions over circuit ensembles are defined precisely as in the classical case. A \textbf{quantum polynomial time algorithm} (or \textbf{QPT} for short) $\mathcal A$ is a uniform ensemble of quantum circuits; algorithms can also include measurements and discarding (or tracing-out) of subsystems, so long as these also admit efficient classical descriptions. The input and output size of a quantum algorithm can vary, and will have to be deduced from context. For example, given a QPT $\mathcal A$, the expression $\Pr[\mathcal A(\ket{0^n}) = 1]$ will take the value zero unless $\mathcal A$ has a specific, labeled output qubit which is measured at the end of the computation.

\section{Quantum encryption}\label{sec:encryption}

In this section, we discuss a recently developed encryption scheme for quantum states, with computational assumptions~\cite{ABFGSS16}. In \expref{Section}{sec:pseudo}, we briefly recall how to construct a classical function which appears pseudorandom to quantum adversaries, by means of a function which is one-way against quantum adversaries. In \expref{Section}{sec:scheme}, we define a notion of symmetric-key quantum encryption, together with associated notions of IND-CPA and IND-CCA1 security. We then describe a scheme which is IND-CCA1-secure under the assumption that quantum-secure one-way functions exist. 

\subsection{Quantum-secure pseudorandomness}\label{sec:pseudo}

We begin with two primitives for encryption: quantum-secure one-way functions, and quantum-secure pseudorandom functions. These are both classical, efficiently computable functions which are in some sense resistant to quantum analysis. In the case of one-way functions, we demand that inversion is hard; in the case of pseudorandom functions, we demand that distinguishing from perfectly random functions is hard.

\begin{definition}\label{def:quantum-secure-owf}
A PT-computable function $f:\{0,1\}^* \rightarrow \{0, 1\}^*$ is a quantum-secure one-way function (qOWF) if for every QPT $\algo A$, 
$$
\emph{Pr}_{x \inrand \{0, 1\}^n} \left[\algo A (f(x), 1^n) \in f^{-1}(f(x)) \right] \leq \negl(n)\,,
$$
where the probability is taken over $x \inrand \{0, 1\}^n$ as well as the measurements of $\algo A$.
\end{definition}

\begin{definition}\label{def:quantum-secure-prf}
A PT-computable function family $f_k : \{0,1\}^n \rightarrow \{0, 1\}^m$ is a quantum-secure pseudorandom function (qPRF) if for every QPT $\algo A$, 
$$
\left|\prob_{k \inrand \{0, 1\}^n} [\algo A^{f_k}(1^n) = 1] - \prob_{g \inrand \mathcal F_{n, m}}[ \algo A^g(1^n) = 1]\right|
\leq \negl(n)\,,
$$
where $\mathcal F_{n, m}$ denotes the space of all functions from $\{0,1\}^n$ to $\{0,1\}^m$.
\end{definition}

Classically, one-way functions are the fundamental primitive underpinning encryption. A series of basic results shows that one-way functions can be turned into pseudorandom functions, which can then be used for defining probabilistic encryption schemes. This series of results carries over to the quantum-secure case without much of a change (although some proofs are somewhat more involved.) For example, it is known how to construct qPRFs from qOWFs.

\begin{theorem}\label{thm:qOWF-implies-qPRF}
If quantum-secure one-way functions exist, then so do quantum-secure pseudorandom functions.
\end{theorem}
\begin{proof} (Sketch.) It is folklore that the well-known H{\aa}stad et al. result that pseudorandom generators can be constructed from any one-way function~\cite{HILL99} carries over to the quantum-secure case. Roughly speaking, the reasoning is that the reduction in the proof is done in a ``black-box'' way, i.e., only by feeding inputs into the adversary and then analyzing the resulting outputs. The quantum-secure case then simply involves replacing PPTs with QPTs in the appropriate places. Proving that the standard GGM construction~\cite{GGM86} of PRFs from pseudorandom generators is still secure in the setting of quantum adversaries is more involved; this was established by Zhandry~\cite{Zhandry2012}.
\end{proof}

\subsection{Symmetric-key encryption of quantum states}\label{sec:scheme}

It is well-known how to encrypt quantum states with information-theoretic security, via the so-called quantum one-time pad. To encrypt a single-qubit state $\rho$, we choose two classical bits at random, use them to select a random Pauli matrix $P \in \{\one, X, Y, Z\}$, and perform $\rho \mapsto P \rho P^\dagger$. To encrypt an $n$-qubit quantum state $\rho$, we select $r \inrand \{0,1\}^{2n}$ and apply
\begin{equation}\label{eq:quantum-one-time-pad}
\rho \longmapsto P_r \rho P_r^\dagger\,,
\end{equation}
where $P_r$ denotes the element of the $n$-qubit Pauli group indexed by $r$. 

One disadvantage of the quantum one-time pad is that parties must share two bits of randomness for every qubit which they wish to transmit securely. In particular, one cannot securely exchange multiple messages with the same key. To address this issue, we must settle for computational security assumptions and use pseudorandomness to select $r$. A general encryption scheme for quantum states is then defined as follows.

\begin{definition}\label{def:encryption-scheme}
A symmetric-key quantum encryption scheme is a triple of QPTs:
\begin{itemize}
\item (key generation) $\KeyGen : 1^n \longmapsto k \in \{0, 1\}^n$;
\item (encryption) $\Enc_k : \states (\mathcal H_m) \longrightarrow \states (\mathcal H_c)$;
\item (decryption) $\Dec_k : \states (\mathcal H_c) \longrightarrow \states (\mathcal H_m)$;
\end{itemize}
where $m$ and $c$ are polynomial functions of $n$, and the QPTs satisfy $\| \Dec_k \circ \Enc_k - \one_m \|_\diamond \leq \negl(n)$ for all $k \in \emph{\supp}\,\KeyGen(1^n)$.
\end{definition}

Public-key quantum encryption schemes are defined in an analogous manner. The encryption schemes we will need must produce ciphertexts which are computationally indistinguishable. In some cases, the ciphertexts will need to remain indistinguishable even to adversaries which possess oracle access to the encryption algorithm (and sometimes also even the decryption algorithm.) This security notion is captured by the following definition.

\begin{definition}\label{def:IND}
A symmetric-key quantum encryption scheme is IND-secure if for all QPTs $\algo A, \algo A'$,
$$
\left|\prob[ (\algo A' \circ \Enc_k \otimes \one_s \circ \algo A) \cdot 1^n = 1] -
\prob[ (\algo A' \circ \Xi_{\Enc_k \ket{0^m}\bra{0^m}} \otimes \one_s \circ \algo A) \cdot 1^n = 1] \right|
\leq \negl(n)\,,
$$ 
where $\Xi_\sigma: \rho \mapsto \sigma$ is the ``forgetful'' map, and $s$ is a polynomial function of $n$. If $\algo A$ and $\algo A'$ have oracle access to $\Enc_k$, then we say that the scheme is IND-CPA secure. If in addition $\algo A'$ has oracle access to $\Dec_k$, then we say that the scheme is IND-CCA1 secure.
\end{definition}

The two QPTs $\algo A$ and $\algo A'$ together model the adversary. The definition above captures the idea of a certain ``security game'' between an adversary and a challenger. The game proceeds in steps: (i.) the key is selected and the adversary receives access to the appropriate oracles, (ii.) after some computation, the adversary transmits the first part of a bipartite state $\rho_{ms}$ to a challenger, (iii.) the challenger either encrypts this or replaces it with the encryption of $\ket{0^m}\bra{0^m}$, and then returns the result to the adversary, and (iv.) the adversary must decide which choice the challenger made. The scheme is considered secure if the adversary can do no better than random guessing. As shown in~\cite{ABFGSS16}, this definition is equivalent to a security notion called \emph{semantic security}; roughly speaking, this notion captures the idea that anyone that tries to compute anything about a plaintext gains no advantage by possessing its encryption. In addition, \expref{Definition}{def:IND} is equivalent to several natural variants, where e.g., the challenger chooses to encrypt one of two messages provided by the adversary, or where the game is played over multiple rounds. The latter guarantees security of transmitting multiple ciphertexts produced via encryption with the same key.

We now show how to use qPRFs to construct simple symmetric-key quantum encryption schemes that satisfy all of the above security conditions.

\begin{theorem}\label{thm:IND-CCA1}
If quantum-secure pseudorandom functions exist, then so do IND-CCA1-secure symmetric-key quantum encryption schemes.
\end{theorem}
\begin{proof}
Let $\{f_k\}$ be a qPRF. For simplicity we assume that each $f_k$ is a map from $\{0, 1\}^n$ to $\{0, 1\}^{2n}$. Recall that for $r \in \{0, 1\}^2n$, $P_r$ denotes the element of the $n$-qubit Pauli group indexed by $r$. Consider the following scheme:
\begin{itemize}
\item $\KeyGen(1^n)$: output $k \inrand \{0, 1\}^n$;
\item $\Enc_k(\rho)$: choose $r \inrand \{0, 1\}^n$; output $\ket{r}\bra{r} \otimes P_{f_k(r)} \rho P_{f_k(r)}^{\dagger}$;
\item $\Dec_k(\ket{r}\bra{r} \otimes \sigma)$: output  $P_{f_k(r)}^\dagger \rho P_{f_k(r)}$\,.
\end{itemize}
In the decryption algorithm, we may assume that the first register is always measured prior to decrypting. Correctness of the scheme is straightforward to check: decrypting with the same key and randomness simply undoes the Pauli operation.

We now sketch the proof that the scheme is IND-CCA1 secure. The key observation is that each query to the encryption oracle is no more useful than receiving a pair $(r, f_k(r))$ for $r \inrand \{0, 1\}^{2n}$, and that each decryption oracle is no more useful than receiving a pair $(r, f_k(r))$ for a string $r$ of the adversary's choice. Thus the adversary learns at most a polynomial number of values of $f_k$. Now, if $f_k$ is a perfectly random function, then these values are completely uncorrelated to the one used to encrypt the challenge. The scheme is thus secure simply by the information-theoretic security of the quantum one-time pad. On the other hand, if $f_k$ is a function in a qPRF, \expref{Definition}{def:quantum-secure-prf} guarantees oracle indistinguishability from perfectly random functions. It follows that, if $(\algo A, \algo A')$ can break the actual scheme, then by computational indistinguishability they would also break the perfect scheme, which is impossible.
\end{proof}

We emphasize that the above proof shows that, even in the case where the adversary chooses the randomness $r$ used by the $\Enc_k$ and $\Dec_k$ oracles, the scheme remains secure. Of course, the randomness for the challenge encryption must still be selected by the challenger. Finally, by combining \expref{Theorem}{thm:qOWF-implies-qPRF} and \expref{Theorem}{thm:IND-CCA1}, we have the following.

\begin{theorem}\label{thm:qOWF-implies-qSKE}
If quantum-secure one-way functions exist, then so do IND-CCA1-secure symmetric-key quantum encryption schemes.
\end{theorem}

\section{Quantum black-box obfuscation}\label{sec:black-box}

In this section, we discuss the virtual black-box framework for obfuscating quantum computations. We begin in \expref{Section}{sec:vbb-definitions} with a definition of black-box quantum obfuscator, motivated both by the classical analogue and an intuitive notion of what a ``good obfuscator'' should achieve. In \expref{Section}{sec:vbb-applications}, we outline several interesting cryptographic consequences that would follow from the existence of such an obfuscator. Finally, in \expref{Section}{vbb:impossibility}, we prove a few impossibility results which restrict the range of possibilities for the existence of black-box quantum obfuscators. Interestingly, our results leave open some possibilities, which include (restricted versions) of the most interesting applications. Indeed, it is conceivable that quantum obfuscation could be significantly more powerful than its classical counterpart.

\subsection{Definitions}\label{sec:vbb-definitions}

Any reasonable notion of obfuscation involves giving the obfuscated circuit $\algo O(C)$ to an untrusted party. We accept as fundamental the idea that this obfuscated circuit should implement some particular, chosen functionality $f_C$, and that the object $\algo O(C)$ allows the untrusted party to execute that functionality. In the black-box formulation of obfuscation, we demand that this is effectively all that the untrusted party will ever be able to do. The rigorous formulation uses the simulation paradigm: anything which can be efficiently learned from the obfuscated circuit, should also be efficiently learnable simply by evaluating $f_C$ some polynomial number of times. This ``virtual black-box'' notion was first formulated by Barak et al.~\cite{BGIRSVY12}, and proved impossible to satisfy generically in the classical case.

In the quantum case, there are several complications. First, we are considering the obfuscation of quantum functionalities. This implies that the end user (and hence also any adversary) should be in possession of a quantum computer, and likewise for the simulator. Second, it is conceivable that the obfuscation may not just be another quantum circuit, which is simply a classical state describing a quantum computation. The obfuscator might instead output a quantum state, which is then to be employed by the end user to execute the desired functionality in some well-specified manner. These considerations motivate the following definition.
\begin{definition}\label{def:vbb-obfuscator}
A \textbf{black-box quantum obfuscator} is a quantum algorithm $\algo O$ and a QPT $\algo J$ such that whenever $C$ is an $n$-qubit quantum circuit, the output of $\algo O$ is an $m$-qubit state $\algo O(C)$ satisfying
\begin{enumerate}
\item (polynomial expansion) $m = \text{poly}(n)$;
\item (functional equivalence) $\bigl\| \algo J ( \algo O(C) \otimes \rho ) - U_C \rho U_C^\dagger \bigr\|_\emph{tr} \leq \negl(n)$ for all $\rho \in \states(\mathcal H_n);$
\item (virtual black-box) for every QPT $\mathcal A$ there exists a QPT $\mathcal S^{U_C}$ such that
$$
\Bigl| \emph{Pr}\bigl[\mathcal A(\mathcal O(C)) = 1\bigr] - \emph{Pr}\bigl[\mathcal S^{U_C}\bigl(\ket{0^n}\bigr) = 1\bigr] \Bigr| \leq \negl(n)\,.
$$
\end{enumerate}
\end{definition}
We emphasize that while the ``interpreter'' algorithm $\algo J$ must be polynomial-time, the obfuscator itself need not be. In applications, it will be necessary to make the obfuscator polynomial-time; on the other hand, our impossibility results will hold even for inefficient obfuscators. One could consider variants of \expref{Definition}{def:vbb-obfuscator} where the interpreter algorithm is fixed once and for all, or where $\algo O(C)$ itself consists of both a quantum ``advice state'' and a circuit which the end user should execute on the advice state and the desired input. It is straightforward to show that all of these variants are equivalent, in the sense that a black-box quantum obfuscator of each variant exists if and only if the other variants exist. Since we are primarily concerned with possibility vs impossibility, we will stick with the formulation in \expref{Definition}{def:vbb-obfuscator}. We also remark that the interpreter is a natural addition to the classical black-box definition when passing to the quantum case. In order for the definition to make sense, there should be \emph{some efficient way} to use $\algo O(C)$ to implement $U_C$; whatever that efficient procedure is, we have here called it an interpreter and denoted it by $\algo J$.

We also point out that the no-cloning theorem opens up the possibility of \emph{computationally unbounded adversaries.} In the classical case, such an adversary could simply execute the circuit on every input, and thus learn far more than is possible for a polynomial-time black-box simulator. Quantumly, however, a computationally unbounded adversary is restricted both by the no-cloning theorem and the limitations of measurement. The adversary may not be able to acquire multiple copies of the obfuscated state, and the single state may be partially (or completely) destroyed when measured. It is thus not \emph{a priori} clear that an unbounded adversary could always outmatch a polynomial-time black-box simulator. The appropriate definition is a straightforward modification of \expref{Definition}{def:vbb-obfuscator}, where we replace the third condition with the following:

\begin{enumerate}
\setcounter{enumi}{2}
\item \emph{(information-theoretic virtual black-box) for every quantum adversary $\mathcal A$ there exists a QPT $\mathcal S^{U_C}$ such that}
$$
\Bigl| \prob \bigl[\mathcal A(\mathcal O(C)) = 1\bigr] - \prob \bigl[\mathcal S^{U_C}\bigl(\ket{0^n}\bigr) = 1\bigr] \Bigr| \leq \negl(n)\,.
$$
\end{enumerate}

\subsection{Applications of efficient black-box obfuscators}\label{sec:vbb-applications}

In this section, we motivate the study of quantum black-box obfuscation by giving a few example applications. Unsurprisingly, these applications require that the obfuscation algorithm is itself quantum polynomial-time; strictly speaking, this is not required of \expref{Definition}{def:vbb-obfuscator}. Many of these applications are motivated by known classical applications of classical black-box obfuscators. Although our impossibility results will put some restrictions on these applications, they remain interesting. In fact, some of the applications (such as quantum-secure one-way functions) will be used in the impossibility proofs themselves. We point out that, while most of the applications below are written in terms of quantum functionality (e.g., encryption of quantum states), one can just as well consider the weaker case of classical functionality, in this case achieved via quantum means (e.g., via a quantum algorithm for obfuscation.)

\subsubsection{Quantum-secure one-way functions}

The first application shows that, if there exists a classical algorithm for obfuscating quantum computations, then quantum-secure one-way functions exist. By the results discussed in \expref{Section}{sec:encryption}, this also implies the existence of quantum-secure pseudorandom generators, quantum-secure pseudorandom functions, and IND-CCA1-secure symmetric-key quantum encryption schemes.

\begin{prop}
If there exists a classical probabilistic algorithm which is a quantum black-box obfuscator, then quantum-secure one-way functions exist.
\end{prop}
\begin{proof}
The proof is essentially the same as that of Lemma 3.8 in \cite{BGIRSVY12}. For all $a \in \{0, 1\}^n$ and $b \in \{0, 1\}$, we define 
$$
U_{a, b} : \ket{x,\, y} \longmapsto
\begin{cases}
\ket{a,\, y \oplus b} &\text{ if } x = a;\\
\ket{x,\, y} &\text{ otherwise}.
\end{cases}
$$
Define a function $f : \{0, 1\}^* \rightarrow \{0, 1\}^*$ by $f(a, b, r) = \algo O_r(U_{a, b})$ where $\algo O$ is the obfuscator\footnote{For simplicity of notation, we omit $\algo J$ and assume that $f(a, b, r) = \algo O_r(U_{a, b})$ is in fact a classical circuit for $U_{a, b}$.} as in the hypothesis, and $\algo O_r$ denotes the same algorithm, but with randomness coins initialized to $r$. Clearly, inverting $f$ requires computing $b$ from $\algo O_r(U_{a, b})$. Moreover, with only black-box access to $U_{a, b}$ (for uniformly random $a, b$) the probability of correctly outputing $b$ in polynomial time is at most $1/2 + \negl(n)$.  By the black-box property of $\mathcal O$, we then have
\begin{align*}
\prob_{a, b} [ A(f(a, b, r)) = b] 
&= \prob_{a, b} [ A(\mathcal O_r(a, b)) = b ]\\
&\leq \prob_{a, b} \left[ S^{U_{a, b}}(1^n) = b\right] + \negl(n)\\
&\leq \frac{1}{2} + \negl(n)\,,
\end{align*}
which completes the proof.
\end{proof}

We remark that the above proof fails if the obfuscator is a quantum algorithm---even if its output is itself classical. The issue is that one-way functions must be deterministic; while one can turn a classical probabilistic algorithm into a deterministic one by making the coins part of the input, this is not possible quantumly. We leave the problem of constructing cryptographically useful primitives from a fully quantum obfuscator (or even just from a quantum encryption scheme) as an interesting open question. 

\subsubsection{CPA-secure private-key quantum encryption}

Can we say anything about encryption of data if we know that \emph{quantum} algorithms for quantum black-box obfuscation exist? While we do not know how to extract one-way functions, we can nonetheless produce useful encryption schemes, as follows.

\begin{prop}\label{prop:PKE-from-SKE}
If quantum black-box obfuscators exist, then so do IND-CPA-secure symmetric-key quantum encryption schemes.
\end{prop}
\begin{proof} (Sketch.)
Let $(\algo O, \algo J)$ be a quantum black-box obfuscator. We consider an adaptation of the unitary operator $U_{a, b}$ defined above, but now with Pauli group action instead of XOR, and with two $n$-bit registers:
$$
U'_{r, k} : \ket{x,\, y} \longmapsto
\begin{cases}
\ket{x,\, P_r^\dagger y} &\text{ if } x = k;\\
\ket{x,\, y} &\text{ otherwise},
\end{cases}
$$
Now consider the following scheme for encrypting $n$-qubit quantum states.
\begin{itemize}
\item $\KeyGen(1^n)$: output $k \inrand \{0, 1\}^n$;
\item $\Enc_{k}(\rho)$: choose $r \inrand \{0, 1\}^n$; output $P_r \rho P_r^\dagger \otimes \algo O(U_{r, k})$;
\item $\Dec_{k}(\sigma \otimes \tau)$: output  the second register of $\algo J(\tau \otimes \ket{k}\bra{k} \otimes \sigma)$.
\end{itemize}
To check correctness, we apply the functionality-preserving property of the obfuscator. A decryption of a valid encryption with the same key yields
\begin{align*}
\Dec_k (\Enc_k (\rho))  
&= \tr_1 \left[\algo J(\algo O(U_{r, k}) \otimes \ket{k}\bra{k} \otimes P_r \rho P_r^\dagger)\right] \\
&= \tr_1 \left[U_{r, k} (\ket{k}\bra{k} \otimes P_r \rho P_r^\dagger) U_{r, k}^\dagger\right] \\
&= \tr_1 \left[\ket{k}\bra{k} \otimes \rho \right] \\
& = \rho\,.
\end{align*}
as desired. IND-CPA security follows from the black-box property of the obfuscator, as follows. Let $\algo A$ be an adversary with access to the encryption oracle. Since the output of the encryption is a product state, $\algo A$ can be simulated by an adversary $\algo S$ that has only the first register of the ciphertext (i.e., $P_r \rho P_r^\dagger$) and black-box access to the unitary $U'_{r, k}$. It's then clear that $\algo S$ can only succeed in the challenge stage of \expref{Definition}{def:IND} by discovering the secret input for $U'_{r, k}$ or by guessing the response to the challenge. In any case, $\algo S$ (and hence also $\algo A$) succeeds with probability at most $1/2 + \negl(n)$.
\end{proof}

\subsubsection{Public-key encryption from private-key encryption}

As we now show, combining black-box obfuscation with one-way functions yields even stronger encryption functionality.

\begin{prop}\label{prop:PKE}
If quantum black-box obfuscators and quantum-secure one-way functions exist, then so do IND-CPA-secure public-key quantum encryption schemes.
\end{prop}
\begin{proof} (Sketch.)
Under the hypothesis, \expref{Theorem}{thm:qOWF-implies-qSKE} implies the existence of IND-CCA1-secure symmetric-key encryption schemes for quantum states. Let $(\KeyGen, \Enc, \Dec)$ be such a scheme; for concreteness, we may take the scheme described in \expref{Theorem}{thm:IND-CCA1}. For $x \in \{0, 1\}^n$, let $\Enc_{(x)}$ denote the encryption circuit for key $x$; this is the circuit that accepts two input registers (one for randomness, and one for the plaintext) and outputs the ciphertext. Now define a public-key encryption scheme $(\KeyGen', \Enc', \Dec')$ as follows.
\begin{itemize}
\item $\KeyGen'(1^n)$: output $sk : = k \inrand \{0, 1\}^n$ (secret key) and $pk := \algo O\left(\Enc_{(sk)}\right)$ (public key);
\item $\Enc'_{pk}(\rho)$: choose $r \inrand \{0, 1\}^n$; output $pk( \ket{r}\bra{r} \otimes \rho)$;
\item $\Dec'_{sk}(\sigma)$: output  $\Dec_{sk}(\sigma)$\,.
\end{itemize}
The correctness of this scheme follows directly from the functionality-preserving property of $\algo O$ and the correctness of the private-key scheme. To prove IND-CPA security for the public-key scheme, we rely on the black-box property. It implies that any QPT adversary $\algo A$ with access to the public key can be simulated by a QPT $\algo S$ having only black-box access to $\Enc_{(sk)}$. The QPT $\algo S$, in turn, can be simulated by a QPT $\algo S'$ which has both decryption and encryption oracles for the private-key scheme $(\KeyGen, \Enc, \Dec)$. It may not be immediately obvious that the decryption oracle is necessary; this is the case because black-box access to $\Enc_{(sk)}$ enables $\algo S$ to select the randomness used for encryption, thus gaining the ability to evaluate pairs $(r, f_{sk}(r))$ where $f$ is the qPRF from the private-key scheme. 

Now we have that, if $\algo A$ can distinguish ciphertexts during the challenge, then so can $\algo S'$; since the ciphertexts themselves are the same for the public-key scheme and the private-key scheme, this contradicts the IND-CCA1 security of the private-key scheme.
\end{proof}

A few remarks are in order. First, in~\cite{ABFGSS16} it is shown that IND-CPA-secure public-key quantum encryption schemes exist under the assumption that quantum-secure trapdoor permutations exist. This is a stronger assumption than one-way functions. \expref{Proposition}{prop:PKE} can then be thought of as replacing this strengthening of assumptions with an obfuscator. In~\cite{BJ15} it is shown how to use quantum-secure classical public-key encryption to produce quantum public-key encryption (by encrypting the key for the quantum one-time pad); this amounts to the same assumption on primitives as in~\cite{ABFGSS16}. An important difference between~\cite{ABFGSS16, BJ15} and \expref{Proposition}{prop:PKE} is that the scheme from \expref{Proposition}{prop:PKE} may have public keys which are quantum states. Such schemes have not been considered before, and (due to no-cloning) would have significantly different features from their classical counterparts.

An interesting question is if there could be public-key encryption for classical data with classical ciphertexts, but where the encryption procedure is performed by a quantum algorithm. While this question remains open, our impossibility results will show that this cannot be achieved in a generic way via \expref{Proposition}{prop:PKE}.

\subsubsection{Quantum fully homomorphic encryption}

We briefly recall the idea of fully homomorphic encryption (FHE). For thorough definitions and the appropriate notions of security in the fully quantum case, see~\cite{BJ15}. Without considering all of the details, we will view QFHE as an encryption scheme (just as in \expref{Definition}{def:encryption-scheme}), but where $\KeyGen$ produces an extra ``evaluation'' key $k_{\opn{eval}}$, and there is an ``evaluation'' algorithm:
\begin{itemize}
\item $\Eval_{k_{\opn{eval}}} : \states (\mathcal H_m \otimes \mathcal H_g) \longrightarrow \states (\mathcal H_m)$.
\end{itemize}
We imagine a party (henceforth, \emph{server}) in possession of $k_{\opn{eval}}$ and a ciphertext $\Enc_k(\rho)$ provided by another party (henceforth, \emph{client}.) The evaluation algorithm then enables the server to produce the ciphertext $\Enc_k(G_k \rho G_k^\dagger)$, where $G$ is a gate of the server's choice. A classical string describing the choice of gate $G$ (and which qubits $k, k+1, \dots$ of $\rho$ it should be applied to) is input into the register $\mathcal H_g$. In general, we may consider the case where $k_{\opn{eval}}$ is itself a quantum state. Depending on the details of the scheme, this key may be partly or fully consumed by $\Eval$; indeed, this is the case in~\cite{BJ15}. Depending on the consumption rate, this might violate the (classically standard) \emph{compactness} requirement for FHE, namely that the amount of communication between the client and the server should scale only with the size of the ciphertext, and not with the size of the computation the server wishes to perform.

\begin{prop}\label{prop:QFHE}
If quantum black-box obfuscators and one-way functions exist, then so do IND-CPA-secure quantum fully homomorphic encryption schemes.
\end{prop}
\begin{proof} (Sketch.)
We will consider the public-key case, which turns out to be simpler. Let $(\algo O, \algo J)$ be a quantum obfuscator, and $(\KeyGen, \Enc, \Dec)$ an IND-CPA-secure public-key scheme. We adapt $\KeyGen$ to produce an evaluation key, and describe the evaluation algorithm. We will require a universal circuit $U_\mu$ for performing gates on $m$-qubit states; this circuit accepts two inputs: an $m$-qubit state, and a description of a gate and indices of the qubits to which the gate should be applied. In our usage, $m$ will be the number of qubits of the ciphertext state.
\begin{itemize}
\item $\KeyGen'(1^n)$: output $\KeyGen(1^n) = (sk, pk)$ and $k_{\opn{eval}} = \mathcal O(\Enc_{pk} \circ U_\mu \circ \Dec_{sk})$;
\item $\Eval_{k_{\opn{eval}}} : \rho \otimes \ket{G}\bra{G} \longmapsto \algo J (k_{\opn{eval}} \otimes \rho \otimes \ket{G}\bra{G})$.
\end{itemize}
where $\ket{G}\bra{G}$ is again just a classical string instructing $U_\mu$ to apply the desired gate. A circuit for $\Enc_{pk} \circ U_\mu \circ \Dec_{sk}$ is given below; the gate register is represented by the bottom wire.\\ 
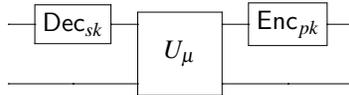
$$
\Qcircuit @C=1em @R=1.0em {
&\gate{\Dec_{sk}} 	&\multigate{1}{~~U_\mu~~} 	&\gate{\Enc_{pk}} 	&\qw\\
&\qw 			&\ghost{~~U_\mu~~}		&\qw 			&\qw\\ \\
}
$$
We now want to show that $(\KeyGen', \Enc, \Dec, \Eval)$ is a public-key QFHE scheme. The homomorphic property follows directly from the definition of $\Eval$ and the functionality-preserving property of the obfuscator. The security of the encryption scheme follows from IND-CPA security of $(\KeyGen, \Enc, \Dec)$ and the black-box property of $(\algo O, \algo J)$. The black-box property implies that each execution of the $\Eval$ algorithm is no more useful than providing the server with an encryption of $G\rho G^\dagger$. However, in the IND-CPA setting, the adversary can already use the CPA oracle to produce encryptions of \emph{arbitrary} plaintexts of her choice (as opposed to just ones which are modifications of the plaintext provided by the client.) There is one additional wrinkle: by repeatedly applying gates (or even just the identity), the adversary can also produce multiple encryptions during the challenge round. However, as shown in~\cite{BJ15}, single-message IND-CPA is equivalent to multiple-message IND-CPA. By the assumption that $(\KeyGen, \Enc, \Dec)$ is IND-CPA secure, it follows that the homomorphic scheme is also secure.

We remark that, in general, the encryption procedure $\Enc_{pk}$ may require an external source of randomness. This is certainly the case in classical encryption, but may not be required if the $\Enc$ algorithm is allowed to perform measurements. In any case, since we are starting with an IND-CPA public-key scheme, the adversary already has access to the public key and the ability to encrypt with randomness of her choice; the ability to choose randomness in $\Eval$ is of no additional benefit.
\end{proof}

\subsubsection{Public-key quantum money}

\paragraph{Quantum money.}
The idea of ``quantum money'' first arose in work by Wiesner~\cite{Wie1983}. The core idea is simple: use a quantum state for representing currency in such a way that the no-cloning theorem of quantum mechanics prevents counterfeiting. These ideas were refined and developed further in several works~\cite{Aar09, AC12, BBBW83, FGHLS12, MS10}; some of these works also included explicit proposals based on various hardness assumptions. 

Informally, a \emph{quantum money scheme} consists of two algorithms: \Mint, which produces quantum states, and \Verify, which accepts an input state and then either accepts or rejects. If the different states produced by \Mint are distinguishable, then we refer to them as \emph{bills}; if they are indistinguishable, then we call them \emph{tokens} (if \Verify consumes them) or \emph{coins} (if \Verify does not consume them.) In all quantum money schemes, we imagine an authority (typically called the bank) which runs \Mint repeatedly to produce money; in addition, the \Verify algorithm should accept only on states produced by the bank. Depending on the particular scheme, this might only be true if \Verify is executed by the bank (private-key money), or it might be true for any party (public-key money.)

In this language, Wiesner's original idea~\cite{Wie1983} was for a private-key scheme for bills, which is as follows. Each execution of \Mint produces two random classical bitstrings $r, s \in \{0, 1\}^{2n}$ as well as an $n$-qubit quantum state $\ket{\psi_r}$, with each qubit initialized in one of the states $\ket{0}, \ket{1}, \ket{+}, \ket{-}$, as determined by the bits of $r$. The bank records the pair $(r, s)$ in a secret table, and publishes $(s, \ket{\psi_r})$. The bank verifies by using $s$ to look up the correct $r$ in the table, and then performing the measurements in the correct basis and checking the results against $r$. 

\paragraph{Public-key money from circuit obfuscation.}
While private-key money schemes are relatively straightforward to construct, public-key proposals appear to be much more difficult, and require computational assumptions. In analogy to its role in producing public-key encryption schemes from private-key ones (\expref{Proposition}{prop:PKE-from-SKE}), an obfuscator can sometimes be used to turn private-key money schemes to public-key ones. The use of an obfuscator to create a particular quantum money scheme was considered by Mosca and Stebila~\cite{MS10}. Their scheme (in our language) is as follows. Each execution of \Mint produces a Haar-random $n$-qubit quantum state $\ket{\psi}$, together with the obfuscation $\mathcal O(U_\psi)$ of a circuit\footnote{For most $\ket{\psi}$, the circuit $U_\psi$ will not have polynomial length. However, as pointed out by~\cite{Aar09}, one can instead select $\ket{\psi}$ from an approximate $t$-design without a significant loss in security.} for $U_\psi = \one - 2\ket{\psi}\bra{\psi}$. The bill consists of the pair $(\mathcal O(U_\psi), \ket{\psi})$. $\Verify(\ket{\varphi})$ consists of executing the following:
$$
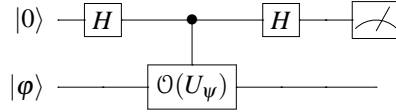
\Qcircuit @C=1em @R=1.0em {
\lstick{\ket{0}} &\gate{H} 	&\ctrl{1}				&\gate{H} 	&\qw &\meter\\
\lstick{\ket{\varphi}} &\qw 	&\gate{\algo O(U_\psi)}	&\qw 	&\qw &\qw\\ \\
}
$$
and accepting iff the measurement returns $1$. It's easy to check that the above succeeds only on valid states; moreover, in that case, the state $\ket{\psi}$ is output in the second register, so that verification can be repeated. To show resistance of the above scheme to counterfeiting, one can use Aaronson's Complexity-Theoretic No-Cloning Theorem~\cite{Aar09}, which states that cloning the state $\ket{\psi}$ while in possession of oracle access to $\ket{U_\psi}$ requires $\Omega(2^{n/2})$ queries. The first published proof of this theorem (as well as its first appearance in the form required here) was in~\cite{AC12}. 

Unfortunately, we will later show that obfuscation of quantum circuits in the form required by Mosca and Stebila is impossible. What remains possible is a setting in which both $\ket{\psi}$ and $\mathcal O(U_\psi)$ are quantum states, and another circuit (which is publicly known and independent of $\ket{\psi}$) is used for verification. Moreover, as we will also show, any black-box obfuscation scheme which outputs states that can be efficiently cloned is also impossible. We thus conjecture the following.
\begin{conjecture}
If quantum black-box obfuscators exist, then so do public-key quantum money schemes.
\end{conjecture}
If the relevant obfuscation is a consumable state, then this would result in a token scheme. If it can be reused to perform verification repeatedly\footnote{For example, if successful verification also outputs another state which is sufficiently close to the original state.}, then the result would be a bills scheme. We remark that, in any case, all of the public-key money states discussed above should be authenticated by the bank; otherwise a merchant would only know that he was handed \emph{some} pair (state, circuit) where the circuit executed on the state outputs ``accept''---a clearly inadequate state of affairs.

\subsection{Impossibility results}\label{vbb:impossibility}

\subsubsection{Impossibility of two-circuit obfuscation}\label{sec-twocircuit}
Barak et. al. \cite{BGIRSVY12} showed that black-box obfuscation is impossible by constructing an explicit circuit family that cannot be black-box obfuscated. 
We begin with a similar result in the quantum setting. We show that quantum black-box obfuscation is impossible in any setting where the adversary can gain access to two outputs of the obfuscator on \emph{different} inputs. We formalize this notion by defining a ``black-box two-circuit obfuscator,''  defined just as in Definition \ref{def:vbb-obfuscator} but with the following strengthening of the virtual black-box condition:
\begin{enumerate}
\setcounter{enumi}{2}
\item \emph{(two-circuit virtual black-box) for every pair of quantum circuits $C_1$ and $C_2$ and every quantum adversary $\mathcal A$ there exists a quantum simulator $\mathcal S^{U_{C_1}, U_{C_2}}$ and a negligible $\epsilon_2$ such that}
$$
\Bigl| \text{Pr}[\mathcal A(\mathcal O(C_1) \otimes \mathcal O(C_2)) = 1] - \text{Pr}\bigl[\mathcal S^{U_{C_1}, U_{C_2}}\bigl(\ket{0}^{\otimes |C_1| + |C_2|}\bigr) = 1\bigr] \Bigr| \leq \epsilon_2(n, \min\{|C_1|,|C_2|\})\,.
$$
\end{enumerate}

We now show that there exists a family of circuits which is unobfuscatable under the above definition. We emphasize that our result holds even when the outputs of the obfuscator are quantum states, and even if these states are \emph{single-use only}, i.e., if the interpreter $\mathcal J$ irrevocably destroys the obfuscated state during use. 

We first define a \emph{circuit-pair family} to be an ensemble of distributions over pairs of circuits. More precisely, if $\mathcal C$ is a circuit-pair family, then there exists a Turing machine M which, on input a positive integer parameter $n$ (in unary), outputs a classical description of a pair of circuits $(C_n, D_n)$ drawn at random from some distribution $\mathcal C_n$ on pairs of poly$(n)$-size circuits. If $M$ is polynomial-time, then we say that $\mathcal C$ is a \emph{poly-time circuit-pair family.}

We also define a \emph{state-pair family} analogously. If $\mathcal C'$ is a state-pair family, then there exists a (not necessarily polynomial-time) quantum algorithm which, on input $n$ in unary, outputs a pair of density operators $(\rho_n, \sigma_n)$ drawn at random from some distribution $\mathcal C_n'$ on quantum states on poly$(n)$-many qubits. Given a circuit-pair family $\mathcal C$ and a state-pair family $\mathcal C'$, we say that $\mathcal C'$ is an obfuscation of $\mathcal C$ if there exists a computable map $\mathcal C \rightarrow \mathcal C'$ assigning to each circuit a corresponding state, in a manner that satisfies the two-circuit obfuscation definition above.

With these definitions, we can now state our first impossibility result.

\begin{theorem}\label{thm:pair-impossibility}
There exists a poly-time quantum circuit-pair family $\mathcal C$ such that no state-pair family is an obfuscation of $\mathcal C$.
\end{theorem}
\begin{proof}
Let $(\mathcal O, \mathcal J)$ be a black-box quantum two-circuit obfuscator. 
The poly-time quantum circuit-pair family $\mathcal C$ consists of quantum circuits for implementing the following pairs of unitary operators. Each pair is parameterized by an input size $n$, as well as bitstrings $a, b$ chosen uniformly at random from $\{0, 1\}^n$.
\begin{align}\label{eq:pair-impossible}
U_{a, b} &: \ket{x,\, y} \longmapsto
\begin{cases}
\ket{x,\, y \oplus b} &\text{ if } x = a;\\
\ket{x,\, y} &\text{ otherwise}.
\end{cases}\\
V_{a, b} &: \ket{C,\, z} \longmapsto
\begin{cases}
\ket{C,\, z \oplus 1} &\text{ if } C(a) = b;\\
\ket{C,\, z} &\text{ otherwise}.
\end{cases}
\end{align}
The registers indexed by $x$ and $y$ are of size $n$. The register indexed by $C$ accepts a circuit description (under some fixed encoding), and needs to be able to handle inputs of size $|\algo O(C_{a, b})|$ (i.e. of size equal to the number of qubits in the state $\algo O(C_{a, b})$). Here $C_{a, b}$ is a fixed, explicit poly$(n)$-size circuit for $U_{a, b}$. The second register of $V_{a, b}$ has size one. 

Note that both of these unitaries can be implemented by efficient quantum circuits. We choose some particular set of such circuits, and henceforth denote them by $C_{a,b}$ and $D_{a,b}$, respectively. The idea for the proof is as follows. Consider an adversary $\algo A$ which is ignorant of the randomly selected $a$ and $b$, and consider two scenarios: in the first, $\algo A$ is given access to \emph{any} pair of \emph{circuits} that implement $U_{a, b}$ and $V_{a, b}$; in the second, $\algo A$ only has oracle access to $U_{a, b}$ and $V_{a, b}$. The point is that, in the first case, $\algo A$ can execute $V_{a, b}$ on a circuit for $U_{a, b}$; provided that the latter is not too long, $\algo A$ will achieve something that is impossible to do with only black-box access. Specifically, it is only in the first case that $\algo A$ will be able to tell if the first circuit/oracle implements $U_{a, b}$, or if it has surreptitiously been replaced by the identity operator!

Things are somewhat complicated by the fact that the obfuscator outputs states instead of circuits. We will need to enable $\algo A$ to execute these states on one another. It will thus be necessary to replace $D_{a, b}$ with a related circuit $D_{a, b}'$. Roughly speaking, this circuit will check if its input, when interpreted as a quantum advice state to the algorithm $\algo J$, maps the input $a$ to the output $b$. A precise description follows. First, $D_{a, b}'$ will have three registers: an input register of $m$ qubits, a work register of $2n$ qubits, and an output register of $1$ qubit initialized in the $|0\rangle$ state. When given as input a quantum state $\rho$ on $m$ qubits, it will initialize the first $n$ bits of the work register to $|a\rangle$, then execute the appropriate unitary circuit of $\algo J$ on $\rho \otimes \ket{a}$. Finally, if the output register of the latter computation contains $\ket{b}$, $D_{a, b}'$ will flip the contents of the output register. We remark that, by a simple counting argument over circuits, this occurs for only an exponentially small fraction of possible input states $\rho$.

Recall that the $2n$-qubit identity operator is denoted by $\one_{2n}$, and is implemented by the obvious circuit which we will denote by $I_{2n}$. We observe that, for every QPT algorithm $\mathcal S$ there exists a polynomial $t$ and a negligible $\epsilon_1$ so that:

\begin{equation}\label{eqn11}
\Bigl| \text{Pr}\bigl[\mathcal S^{U_{a,b}, D_{a,b}'}\bigl(\ket{0}^{\otimes t(n)}\bigr) = 1\bigr]
- \text{Pr}\bigl[\mathcal S^{Id_{2n}, D_{a,b}'}\bigl(\ket{0}^{\otimes t(n)}\bigr) = 1\bigr] \Bigr| 
\leq \epsilon_1(n)\,.
\end{equation}
Here the probability is taken over the uniformly random choice of $a$ and $b$ as well as all of the measurement outcomes of $\algo S$. The above is an easy corollary of the tightness of the Grover bound for unstructured quantum search \cite{BBBV}. Indeed, given the definitions of $U_{a, b}$ and $D_{a, b}'$, it's clear that with only polynomial queries and no knowledge of $a$ or $b$, $\algo S$ is faced precisely with unstructured search for an exponentially small ``marked space.'' This marked space is only encountered if $\algo S$ correctly guesses $a$, or correctly guesses an obfuscation of a circuit that maps $a$ to $b$.
  
Now consider the QPT algorithm $\mathcal{A}$ that, given as input the obfuscated states $\mathcal{O}(C)$ and $\mathcal{O}(D)$, simply executes the quantum algorithm $\mathcal{J}$ on their tensor product, accepting if and only if the outcome is $1$. Notice that this succeeds with constant probability $\alpha > 0$ if $C$ is functionally equivalent to $C_{a,b}$ and $D$ is functionally equivalent to $D_{a,b}'$. On the other hand, this same algorithm $\algo A$ accepts with at most negligible probability when $C$ is functionally equivalent to $I_{2n}$ (and $D$ is still functionally equivalent to $D_{a, b}'$); indeed, this only happens is if $a = b$. Thus there exists a negligible function $\epsilon_2$ so that:

\begin{equation}\label{eqn2}
\Bigl| \text{Pr}\bigl[\mathcal A(\mathcal{O}(D'_{a,b}),\mathcal{O}(I_{2n})) = 1\bigr]
- \text{Pr}\bigl[\mathcal A\bigl(\mathcal{O}(D'_{a,b})\otimes\mathcal{O}(C_{a,b})\bigr) = 1\bigr] \Bigr| 
\geq \alpha-\epsilon_2(n)\,.
\end{equation}

To complete the proof, we explicitly define the poly-time circuit-pair family $\mathcal C$. The distribution $\mathcal C_n$ is generated by choosing $a, b$ uniformly at random from $\{0,1\}^n$, and then choosing a bit $r \in {0, 1}$ at random; if $r = 0$, we output the circuit pair $(C_{a,b}, D'_{a,b})$, and otherwise we output $(I_{2n}, D'_{a,b})$. For this distribution, equations \eqref{eqn11} and \eqref{eqn2} together show that no state-pair family is an obfuscation of $\mathcal C$.
\end{proof}

\subsubsection{Generalizing the impossibility result}

Our goal in this section is to extend the two-circuit impossibility proof from the prior section to the case of obfuscating a single circuit. For our impossibility proof, we require an additional condition on the obfuscator: that each of its outputs is reusable a polynomial number of times. This is a natural condition which is automatically satisfied by classical obfuscators (as well as quantum obfuscators with classical outputs), since their outputs can be perfectly copied.

\begin{definition}\label{def:vbb-reusable}
A \textbf{reusable-black-box quantum obfuscator} is a quantum algorithm $\algo O$ and a QPT $\algo J$ such that whenever $C$ is an $n$-qubit quantum circuit, $\mathcal{O}(C)$ is an $m$-qubit quantum state satisfying
\begin{enumerate}
\item (polynomial slowdown)  $m = \text{poly}(n, |C|)$;
\item (functional equivalence) $\bigl\| \algo J ( \mathcal O(C) \otimes \,\cdot\, ) - C \,\cdot\, C^\dagger \bigr\|_\diamond \leq \negl(n, |C|)$;
\item (reusability) after execution of $\algo J$, an output register contains a state which satisfies (2.);
\item (virtual black-box) for every QPT adversary $\mathcal A$ there exists a QPT simulator $\mathcal S^{U_C}$ such that:
$$
\Bigl| \emph{Pr}[\mathcal A(\rho_{(i)}) = 1] - \emph{Pr}\bigl[\mathcal S^{U_C}\bigl(\ket{0}^{\otimes |C|}\bigr) = 1\bigr] \Bigr| \leq \negl(n, |C|)\,.
$$
\end{enumerate}
\end{definition}

We remark that reusability can be achieved in any number of ways: by providing a state which partially survives uses by the interpreter $\algo J$, by providing sufficiently many copies, or by providing a means of cloning the state. We prove impossibility of the above definition in any setting where the adversary receives two copies of the obfuscator output, even on identical inputs. This is automatically satisfied if the obfuscator provides multiple copies in order to satisfy reusability, or if the state is (even approximately) cloneable. The key new obstacle is to prove impossibility even though the functionality for both copies is \emph{the same.}

To state the result, we define (in analogy to circuit-pair families and state-pair families) a \emph{circuit family} to be an ensemble of distributions over circuits, and a \emph{state family} to be an ensemble of distributions over states. A state family $\mathcal C'$ is said to be an obfuscation of a circuit family $\mathcal C$ if there exists a computable map $\mathcal C \rightarrow \mathcal C'$ assigning to each circuit a corresponding state, in a manner that satisfies \expref{Definition}{def:vbb-reusable}. With these definitions, we will prove the following theorem.

\begin{theorem}\label{thm:single-impossibility}
If quantum-secure one-way functions exist, then there exists a quantum circuit family $\mathcal C$ such that no state family is a reusable-black-box quantum obfuscation of $\mathcal C$.
\end{theorem}

Since the full proof of \expref{Theorem}{thm:single-impossibility} is somewhat lengthy and involved, we will first prove a simpler case, showing that quantum circuits cannot be obfuscated into quantum circuits, under any of the definitions considered so far---even the strongest one, \expref{Definition}{def:vbb-obfuscator}.) This corollary (stated below as \expref{Theorem}{thm:cor}) is arguably the most direct quantum generalization of the impossibility result of ~\cite{BGIRSVY01}. Once we have proved it, we will explain in detail how the proof should be adapted in order to achieve \expref{Theorem}{thm:single-impossibility}.

\begin{theorem}\label{thm:cor}
If quantum-secure one-way functions exist, then there exists a quantum circuit family $\mathcal C$ such that no quantum circuit family is a black-box obfuscation of $\mathcal C$.
\end{theorem}

\begin{proof}
Let $\mathcal O$ be a black-box quantum obfuscator satisfying \expref{Definition}{def:vbb-obfuscator}, such that its outputs are classical bitstrings. Since these states are used to describe an efficiently implementable quantum computation, we can assume that these bitstrings are in fact quantum circuits under some particular encoding. 

To construct the unobfuscatable circuit family, we will need a notion of combining the functionality of two quantum circuits into one. 
\begin{definition} The {\bf combined quantum circuit} of a finite collection $\lbrace C_1,C_2,...,C_k \rbrace$ of $n$-qubit quantum circuits is the circuit that has two registers (a control register of $\log{k}$ qubits, and an input register of $n$ qubits) and, controlled on the value of the first register, applies the respective quantum circuit to the input register.   
\end{definition}

\noindent Notice that if each circuit $C_i$ in the collection is polynomial size, and $k$ is bounded by a polynomial in $n$, then the associated combined quantum circuit is also of polynomial size. We will denote the operation of combining circuits with $\#$. For example, the combined circuit of two circuits $C_1$ and $C_2$ is denoted $C_1 \# C_2$. 

Now recall the two circuits  $C_{a,b}$ and $D_{a,b}$ from Section \ref{sec-twocircuit}, as well as the circuit $I_{2n}$, which simply implements the identity operator on $2n$ qubits. Consider the combined quantum circuits $C_{a,b}\#D_{a,b}$ and $I_{2n}\#D_{a,b}$, sampled by selecting $a$ and $b$ uniformly at random from $\{0, 1\}^n$. We again choose $C = C_{a, b}$ or $C = I_{2n}$, each with probability $1/2$, and ask the adversary to determine which is the case. Using the same reasoning as in the proof of \expref{Theorem}{thm:pair-impossibility} from Section \ref{sec-twocircuit}, these combined quantum circuit distributions are indistinguishable from the perspective of any QPT simulator that is ignorant of $a$ and $b$, and is given only black-box access to $U_{C \# D_{a, b}}$.  On the other hand, unlike in the prior proof, it is not immediately apparent how to distinguish the two possibilities given a circuit description of $\algo{O}(C \#D_{a,b})$. Still, the idea is simple. Suppose we have two copies of the circuit. We can hard-wire the control register of one copy to implement $D_{a, b}$, and hard-wire the control register of the other copy to implement $C$. If we then run the first copy on the second, the result should be equivalent to implementing $D_{a, b}$ on input $C$, which will determine the nature of $C$ and conclude the proof just as in \expref{Theorem}{thm:pair-impossibility}.

Unfortunately, this idea does not work as stated, because the two circuit copies have the same size. Since they are functionally nontrivial, their input size is much smaller than their description, making it impossible to run one on the other. The core difficulty is that $D_{a,b}$ cannot be made large enough to universally execute circuits of size $|\algo O(C \# D_{a, b})| > p(|D_{a, b}|)$ where $p$ is the running time of $\algo O$. This issue arises in the classical proof as well, and is resolved as follows. First, note that we could have $D_{a, b}$ simply provide $a$ and $b$, thus offloading the gate-by-gate execution of $C$ to the algorithm $\algo A$ in the black-box definition. Unfortunately, this would also provide the simulator $\algo S$ with $a$ and $b$, enabling it to simulate $\algo A$. The resolution is to have $D_{a, b}$ provide \emph{encryptions of} $a$ and $b$, as well as a \emph{quantum fully-homomorphic encryption (QFHE) oracle} for homomorphically applying the gates of $C$. We emphasize that the functionality and security of the QFHE oracle crucially depends on the obfuscation property; in particular, an actual QFHE scheme is not required for the proof.

Concretely, we prove the following Lemma, which is a quantum analogue of Lemma 3.6 (including Claim 3.6.1) from~\cite{BGIRSVY12}.

\begin{lemma}\label{lemma-circuitdistribution}
If quantum-secure one-way functions exist, then for each $n \in \mathbb{N}$ and $a, b\in\{0,1\}^n$, there exists a distribution $\mathcal{D}_{a,b}$ over circuits such that:
\begin{enumerate}
\item There exists a PPT algorithm that, given $n\in\mathbb{N}$ and $a,b\in\{0,1\}^n$, samples from $\mathcal{D}_{a,b}$;
\item There is a QPT algorithm $\algo{A}$ so that $C \ket{a} \ket{0^n} = \ket{a}\ket{b}$ implies $\algo{A}^{U_D}(C, 1^n) = a$; this holds for all $n \in \mathbb{N}$, all $a,b\in\{0,1\}^n$, any $D\in \supp(\mathcal{D}_{a,b})$, and every $n$-qubit circuit $C$;
\item{For any QPT $S$, $\Pr[S^{U_D}(1^n)=a] \leq \negl(n)$, where the probability is over $a,b\in\{0,1\}^n$, $D\sim \mathcal{D}_{a,b}$, and the measurements of $S$.}	
\end{enumerate}
\end{lemma}

\begin{proof} The distribution $\mathcal D_{a, b}$ will be sampled by choosing $k, r \inrand \{0, 1\}^{2n}$. The bitstring $k$ is to serve as a private key for the IND-CCA1-secure symmetric-key quantum encryption scheme from \expref{Theorem}{thm:IND-CCA1}. Each circuit $D\in \supp(\mathcal{D}_{a,b})$ will be a combination (again via \#) of the following three circuits. 
\begin{enumerate}
\item $E_{K,a}$; this outputs $\Enc_k(\ket{a})$, executed with randomness $r$. 
\item $\Homorcl_K(G,\rho)$; on input a gate description $G$ and a state $\rho$, it outputs $[\Enc_k \circ G \circ \Dec_k](\rho)$.
\item $B_{k,a,b}$; on input $\rho$, it outputs $\ket{a}$ if $\Dec_{k}(\rho)= \ket{b}$ and $\ket{0^n}$ otherwise.
\end{enumerate}
We remark that the $\Homorcl$ oracle requires randomness in order to re-encrypt the state. This is handled in the usual way: we expand the input in some register via a (quantum-secure) pseudo-random function; these exist by the assumption of quantum-secure one-way functions and \expref{Theorem}{thm:qOWF-implies-qPRF}. 

Clearly, given $a$ and $b$, $\mathcal{D}_{a,b}$ can be sampled efficiently by choosing $k$ uniformly at random and outputting the combined quantum circuit $D_{k, a, b} := E_{k,a}\#\Homorcl_{k}\#B_{k,a,b}$. This establishes Property 1 from the Lemma.  For Property 2, let $\algo{A}$ be the algorithm that, on input a quantum circuit $C$, (i.) uses the first two circuits comprising $D_{k,a,b}$ to homomorphically simulate $C$ gate-by-gate on $a$, and then (ii.) plugs the final state into $B_{K,a,b}$. 

To complete the proof of \expref{Lemma}{lemma-circuitdistribution}, we must verify Property 3, i.e., that no QPT simulator algorithm that has black-box access to  each of the three algorithms comprising $D_{K,a,b}$ can discover $a$ with non-negligible probability.  This amounts to showing that
\begin{equation}\label{eq:hom-gap}
\Bigl| \Pr[\mathcal S^{\Homorcl_k,\Enc_k}(\Enc_k(\ket{0})) = 1] 
- \Pr\bigl[\mathcal S^{\Homorcl_k,\Enc_k}(\Enc_k(\ket{a})) = 1 \Bigr| 
\leq \negl(n)\,,
\end{equation}
where the probabilities are over $k \inrand \{0,1\}^n$ and the measurement outcomes of $\algo{S}$. We proceed by contradiction, and assume that there's a QPT $\algo{S}$ that violates the claim. 

First, we replace the responses to all of $\algo{S}$'s queries to the $\Homorcl_k$ oracle with encryptions of $|0^n\rangle$, and deploy a hybrid argument to show that the success probability of $\algo S$ is not significantly affected. To this end, consider the following hybrids of the computation of $\algo S$ on input $\Enc_k(\ket{a})$: in the $i$-th hybrid, the first $i$ oracle queries of $\algo S$ are answered using $\Homorcl_k$, and the rest are answered with $\Enc_k(|0^n\rangle)$.  Notice that any gap in distinguishing between the $i$ and $i+1$st hybrid must be due to the $i+1$st query $\algo S$ makes to $\Homorcl_k$. We can use this to create a CCA1 adversary $\algo T$ which breaks the encryption scheme, as follows. The QPT $\algo T$ simulates $\algo S$ and replies to its first $i$ oracle queries by means of $\algo T$'s $\Enc_k$ and $\Dec_k$ oracles. Upon receiving the challenge ciphertext, $\algo T$ passes it to $\algo S$ as the response to its $i+1$st oracle query. Finally, $\algo T$ answers the remaining oracle queries of $\algo S$ with $\Enc_K(|0^n\rangle)$. We conclude that, if $\algo S$ violated \eqref{eq:hom-gap}, then $\algo T$ succeeds with non-negligible probability. This establishes (by contradiction) that we can replace the oracle queries of $\algo S$ with $\Enc_k(\ket 0)$. With this replacement, $\algo S$ can distinguish an encryption of $|0^n\rangle$ from an encryption of $\ket{a}$, when given access to only an encryption oracle, which again contradicts IND-CCA1 security of the scheme. 
\end{proof}

Now we are ready to describe the unobfuscatable family of quantum circuits and complete the proof of \expref{Theorem}{thm:cor}. First, for a fixed $a,b \in \{0, 1\}^n$ we let $\mathcal{D}_{a,b}$ be the distribution over circuits constructed in \expref{Lemma}{lemma-circuitdistribution}. Then consider the following two distributions over circuits: 
\begin{enumerate}
\item $\mathcal{F}_{n}$: Choose $a,b \inrand \{0,1\}^n$,  sample a circuit $D_{a, b}$ from $\mathcal{D}_{a,b}$ and output $C_{a,b}\# D_{a,b}$ 
\item $\mathcal{G}_{n}$: Choose $a,b \inrand \{0,1\}^n$,  sample a circuit $D_{a, b}$ from $\mathcal{D}_{a,b}$ and output $I_{2n}\# D_{a,b}$
\end{enumerate}
By Property 2 of \expref{Lemma}{lemma-circuitdistribution} there exists an algorithm $\algo{A}$ that, on input $\mathcal{O}(C_0)$, accepts if $C_0$ was sampled from $\algo{F}_n$ and rejects if $C$ was sampled from $\algo{G}_n$.  Thus there exists a constant $\alpha$ and a negligible function $\epsilon_1$ so that:
$$
\Bigl| \text{Pr}\bigl[\mathcal A(\mathcal{O}(\mathcal{F}_n)) = 1\bigr]
- \text{Pr}\bigl[\mathcal A\bigl(\mathcal{O}(\mathcal{G}_n)\bigr) = 1\bigr] \Bigr| 
\geq \alpha-\epsilon_1(n)\,.
$$
On the other hand, by Property 3 of \expref{Lemma}{lemma-circuitdistribution}, we know that for every QPT $S$ there exists some negligible function $\epsilon_2$ so that: 
$$
\Bigl| \text{Pr}\bigl[\mathcal S^{\mathcal{F}_n}\bigl(\ket{0}^{\otimes n}\bigr) = 1\bigr]
- \text{Pr}\bigl[\mathcal S^{\mathcal{G}_n}\bigl(\ket{0}^{\otimes n}\bigr) = 1\bigr] \Bigr| 
\leq \epsilon_2(n)\,.
$$
We conclude that the circuit family formed by taking the union of $\mathcal F_n$ with $\mathcal G_n$ (and assigning them each equal probability) is an unobfuscatable circuit family.
\end{proof}

We now return to the proof of \expref{Theorem}{thm:single-impossibility}, and show how to extend the above proof to the case where the obfuscator outputs reusable quantum states.

\begin{proof} (of \expref{Theorem}{thm:single-impossibility})
Our proof will still use the same distribution $\mathcal D_{a, b}$ over circuits, which were provided by \expref{Lemma}{lemma-circuitdistribution} and described above, but with some slight modifications. The goal will still be to take two copies of $\algo O(C_0)$ for any circuit $C_0$ sampled from that distribution, and give an algorithm $\algo A$ that can ``execute one copy on the other.'' This will enable us to distinguish if $C_0$ is from the distribution $\mathcal F_n$, or the distribution $\mathcal G_n$ (just as above), a task which is impossible with only black-box access.

However, executing one copy of $\algo O(C_0)$ on another is now somewhat more complicated, due to the fact that we no longer have explicit circuit descriptions in hand, and must instead use the interpreter $\algo J$ (with some register initialized to $\algo O(C_0)$) whenever we want to run $C_0$. To do this, we will need to describe a new distribution $\mathcal D_{a, b}'$ of circuits, closely related to the distribution $\mathcal D_{a, b}$ from \expref{Lemma}{lemma-circuitdistribution}.

\textbf{Attempt 1.~}To warm up, a first attempt at describing $\algo A$ and the modified circuits $D_{a, b}'$ from the distribution $\mathcal D_{a, b}'$ is as follows. First, we simply define $D_{a, b}'$ to be a composition of circuits which simply output both $a$ and $b$. Set $C_0 = C \# D_{a, b}'$, and suppose that inputting $\ket{0}$ in the first register executes the first circuit in the combination, while $\ket{1}$ executes the second circuit in the combination. The algorithm $\algo A$ is in possession of two copies of $\algo O(C_0)$. It performs:
\begin{enumerate}
\item run $\algo J ( \algo O(C_0) \otimes \ket{1} \ket{0^n})$ to retrieve $\ket{a} \ket{b}$ (by functional equivalence of $\algo O$);
\item run $\algo J( \algo O(C_0) \otimes \ket{0} \ket{a})$ (now using the other copy of $\algo O(C_0)$);
\item compare the result to $\ket{b}$.
\end{enumerate}
This does exactly what we want, except of course that the black-box simulator $\algo S$ will also be able to retrieve $a$ and perform this experiment. Our valiant attempt failed.

\textbf{Attempt 2.~}Undeterred, we now try a more sophisticated approach, returning to the idea of encryption and homomorphic execution. We now ask that (as before) $\mathcal D_{a, b}'$ outputs an encryption $\Enc(a)$ of $a$ (when given the flag $A$), implements a $\Homorcl$ oracle (when given the flag $H$), and checks if the input is $\Enc(b)$ (when given the flag $B$). We again set $C_0 = C \# D_{a, b}'$ and give $\algo A$ two copies of $\algo O(C_0)$; for convenience we denote them $\algo O(C_0)$ and $\algo O(C_0)'$. Now, $\algo A$ performs: 
\begin{enumerate}
\item run $\algo J \bigl( \algo O(C_0) \otimes [\ket{1A} \otimes \ket{0^n}]\bigr)$ to retrieve $\Enc(a)$;
\item run $\algo J \bigl( \algo O(C_0) \otimes [\ket{1H} \otimes \algo O(C_0)' \otimes \ket{0} \otimes \Enc(a)]\bigr)$ to ``homomorphically run $C$'';
\item run $\algo J \bigl( \algo O(C_0) \otimes [\ket{1B} \otimes \rho]\bigr)$ where $\rho$ is the output of the previous step; output the result.
\end{enumerate}
The first and last step are largely self-explanatory: we start with the encryption of $a$, and check at the end that we have the encryption of $b$. What happened in the second step? We tried to homomorphically evaluate\footnote{Selecting $C$ was done by passing the bit flag $\ket 0$ into the control register.} $C$ on $\Enc(a)$. By functional equivalence of $\algo O$, we executed the first copy of $C_0$ on input $\ket{1H} \otimes \algo O(C_0)' \otimes \Enc(a)$; this specifies that $D_{a, b}'$ should run the $\Homorcl$ oracle with input $\mathcal O(C_0)' \otimes \Enc(a)$. To make this sensible, we can redefine $\Homorcl$ to accept two registers, and homomorphically evaluate the appropriate circuit of $\algo J$; the result is that, whenever $\Homorcl$ is called on $\mathcal O(C) \otimes \Enc(z)$ for a circuit $C$ and state $z$, the output is $\Enc(C(z))$.

This attempts looks like it succeeds, but there is a disastrous flaw: $\Homorcl$ must now accept inputs with at least as many qubits as $\algo O(C_0)$, which is significantly bigger than the circuit description allowed for $\Homorcl$ itself (since it is a part of $D_{a, b}'$ and hence also of $C_0$).

\textbf{Attempt 3.~} In the final attempt, we will repair the flaw of Attempt 2. The key is to again offload some of the execution, but this time from the $\Homorcl$ oracle to the main algorithm $\algo A$. More precisely, we will expand step (2) in Attempt 2, and execute it gate-by-gate. In this iteration, $\Homorcl$ is back to its original version, and is used only to apply two-qubit gates. It accepts $n$ input qubits, decrypts them, applies the desired gate (as specified in another register), and then re-encrypts. In addition, will also expand $D_{a, b}'$ to provide an $\Enc$ oracle (when given the flag E); we can do this for free, by equation \eqref{eq:hom-gap} in \expref{Lemma}{lemma-circuitdistribution}. The final algorithm $\algo A$ will proceed as follows. Here we have let $m$ denote the number of qubits of $\algo O(C_0)'$, and we let $J_m$ be the circuit of $\algo J$ for executing $m$-qubit obfuscated states.
\begin{enumerate}
\item run $\algo J \bigl( \algo O(C_0) \otimes [\ket{1E} \otimes \algo O(C_0)'_{(k)}]\bigr)$ for all $k \in [m]$, to encrypt all qubits of $\algo O(C_0)'$; 
\item run $\algo J \bigl( \algo O(C_0) \otimes [\ket{1A} \otimes \ket{0^n}]\bigr)$ to retrieve $\Enc(a)$;
\item let $j = 0$ and let $\sigma_j := \Enc(a)$;
\item let $G_j$ be the $j$th gate in the description of $J_m$;
\begin{enumerate}
\item let $s, t$ be the qubits $G_j$ acts on; assume $s$ is a qubit of $\algo O(C_0)'$ and $t$ is a qubit of $\sigma_j$;\footnote{This assumption is only made for simplicity of the algorithm description; the other possibilities are similar.}
\item set $\sigma_{j+1} = \algo J \bigl( \algo O(C_0) \otimes [\ket{1H} \otimes \ket{G_j, s, t} \otimes \Enc(\algo O(C_0)'_{(s)}) \otimes \sigma_j]\bigr)$;\,\footnote{The notation $\algo O(C_0)'_{(s)}$ is meant to indicate that only the $s$-th qubit of that state is to be placed in the input register.}
\item if $j = |J_m|$, continue; otherwise increment $j$ by $1$ and go to step 4.
\end{enumerate}
\item run $\algo J \bigl( \algo O(C_0) \otimes [\ket{1B} \otimes \sigma_j]\bigr)$ and output the result.
\end{enumerate}
A few remarks are in order. First, the reusability of the state $\algo O(C)$ was crucial in our ability to repeatedly execute the $\Homorcl$ oracle in step 4.(b). Second, one checks that by the functional equivalence of the obfuscator and the definition of the $\Homorcl$ oracle, if $C_0 = C_{a, b} \# D_{a, b}'$ then the state $\sigma_j$ when step 5. is reached will indeed be $\Enc(b)$. Third, note that the ``compactness'' issue of $\Homorcl$ from Attempt 2 has been resolved, and the input to $C_0$ in step 4.(b) is now of size $n$ (plus a constant.) 

Finally, despite all of the difficulties in defining the algorithm $\algo A$ appropriately, the hardness of the corresponding search problem for the black-box simulator $\algo S$ is essentially unchanged from the proof of \expref{Theorem}{thm:cor}. The only difference is that $D_{a, b}'$ now also provides an encryption oracle; the encryption scheme we selected is certainly secure in this setting. 

To finish the proof, we again build a circuit family by choosing $C_{a, b} \# D_{a, b}'$ or $I_{2n} \# D_{a, b}'$, each with equal probability, for random $a$ and $b$. By the above arguments, this circuit family is unobfuscatable. This concludes the proof of \expref{Theorem}{thm:single-impossibility}.
\end{proof}

\section{Quantum indistinguishability obfuscation}\label{sec:indistinguishability}

In this section, we analyze a different definition of quantum obfuscation, motivated by classical definitions established by Goldwasser and Rothblum~\cite{GR07}. As opposed to the black-box approach, these definitions assess the quality of an obfuscation in relative terms, e.g., as compared to other functionally-equivalent circuits (or, in the quantum case, states).  

\subsection{CPTP circuits and ensembles of states}

For our discussions on indistinguishability obfuscation, we modify the notations of \expref{Section}{sec:notation} slightly, as follows. 

First, henceforth we take the point of view that any quantum circuit $C$ can include unitary gates \emph{as well as measurement gates and instructions for discarding qubits.} In keeping with this view, we overload notation so that $C$ denotes both the circuit itself (i.e., a classical description of a set of wires and gates) as well as the CPTP map that the circuit implements. 

Second, we allow our ensembles to be indexed by infinite sets of strings rather than the natural numbers. As before, we will denote ensembles by italicized capital letters (and optionally place the indexing set in the subscript.) Elements of the ensemble will be denoted by the corresponding non-italicized capital letter (optionally, with the index in the subscript.) So, for example, we may write $\mathcal C_S := \{C_s : s \in S\}$. Just as in the single-circuit case, we will overload this notation for circuit ensembles and use it to also refer to the corresponding family of CPTP maps. We will need one new definition in this context: given two circuit ensembles $\mathcal C_S$ and $\mathcal D_S$, we say that $\mathcal C_S$ is \textbf{functionally equivalent} to $\mathcal D_S$ (denoted $\mathcal C_S \cong \mathcal D_S$) if the circuits themselves are functionally equivalent, i.e., if $\|C_s - D_s\|_\diamond \leq \negl(|s|)$ for all $s \in S$. 

We will also now need to handle infinite collections of states. We thus define a \textbf{state ensemble} to be an infinite collection $\{\rho_x : x \in X\}$ of density operators indexed by some set $X \subset \{0, 1\}^*$, such that $\rho_x \in \states(\mathcal H_{p(|x|)})$ where $p$ is bounded by some fixed polynomial function. We remark that a circuit ensemble is a special case of a state ensemble, where each state is a classical string describing wires, gates, and so on. A \textbf{uniform state ensemble} will be a state ensemble $\{\rho_x : x \in X\}$ together with a uniform circuit ensemble $\{C_x : x \in X\}$ such that $\rho_x = C_x |0^m\rangle$ for appropriately chosen $m$. We remark that a QPT is defined by choosing a uniform circuit ensemble, and that the set of possible outputs of a QPT are a uniform state ensemble. In particular, if $\mathcal S$ is a state ensemble and $\algo A$ is a QPT, then (ignoring some uninteresting bookkeeping) we may write $\algo A(\mathcal S)$ to denote the state ensemble that results from running $\algo A$ on inputs from $\mathcal S$. Note that if $\mathcal S$ is uniform, then $\algo A(\mathcal S)$ is also uniform.

Next, we consider three different notions of distinguishability for state ensembles, in order of decreasing power. We will write $\mathcal R \approx \mathcal S$ to indicate that $\mathcal R$ and $\mathcal S$ are indistinguishable state ensembles; the type of indistinguishability should be clear from context. Let $\mathcal R = \{\rho_x : x \in X\}$ and $\mathcal S = \{\sigma_x : x \in X\}$ be two state ensembles. We say that $\mathcal R$ and $\mathcal S$ are \textbf{perfectly indistinguishable} if $\rho_x = \sigma_x$ for every $x \in X$. While this is an unnatural notion of indistinguishability for general quantum states, it is more reasonable (and is easy to test for individual cases) if $\rho_x$ and $\sigma_x$ happen to all be classical. A weaker notion is \textbf{statistical indistinguishability}, which demands that $\|\rho_x - \sigma_x\|_\text{tr} \leq \poly(|x|)$ for all but finitely many $x \in X$. In the weakest notion we will consider, we say that $\mathcal R$ and $\mathcal S$ are \textbf{computationally indistinguishable} if for every QPT $\algo A$,
$$
\left|\Pr[\algo A(\rho_x) = 1] - \Pr[\algo A(\sigma_x) = 1]\right| \leq \negl(|x|)
$$
for all but finitely many $x$; here the probabilities are taken over the coins and measurements of $\algo A$. One may also consider a non-uniform version of the above definition, in which $\algo A$ also ranges over non-uniform circuit families, and is allowed access to an auxiliary state ensemble $\{\xi_x : x \in X\}$. All of our results hold (with appropriate adjustments) for both the uniform and the non-uniform setting; we will focus on the uniform setting for convenience. 

\subsection{Definitions: indistinguishability, best-possible}

We begin with the notion of a quantum indistinguishability obfuscator. As before, the interpreter and the obfuscated state must be efficient, while the obfuscation algorithm itself might not be. We also assume that all ``interpreter'' algorithms $\mathcal J$ have two registers: an advice register (where the obfuscated state is to be inserted), and an input register (where the input is to be inserted.) It will thus be convenient to write, e.g., $\mathcal J_{\rho}$ for the CPTP map family defined by the circuits of $\mathcal J$ with the advice register initialized to the state $\rho$.

\begin{definition}\label{def:translator}
A \textbf{quantum translator} is a quantum algorithm $\mathcal O$ and a QPT $\mathcal J$ such that whenever $\mathcal C$ is a circuit ensemble and $C \in \mathcal C$ is an $n$-qubit circuit, 
\begin{enumerate}
\item (polynomial slowdown) $\mathcal O(C)$ has at most $\poly(n)$ qubits;
\item (functional equivalence) $\bigl\| \algo J_{\mathcal O(C)} - C \bigr\|_\diamond \leq \negl(n)$;
\end{enumerate}
\end{definition}

\begin{definition}\label{def:indistinguishability}
A \textbf{quantum statistical (resp., computational) indistinguishability obfuscator} is a quantum translator $(\mathcal O, \mathcal J)$ such that whenever $\mathcal C_S$ and $\mathcal D_S$ are functionally-equivalent circuit ensembles with $|C_s| = |D_s|$ for all $s \in S$, then $\mathcal O(\mathcal C_S)$ and $\mathcal O(\mathcal D_S)$ are statistically (resp., computationally) indistinguishable.
\end{definition}

Note that an obfuscator which simply outputs circuit descriptions is included as a special case of the above; in that case, $\mathcal O(C)$ is always a quantum circuit, and $\mathcal J$ is a universal circuit which executes $\mathcal O(C)$ on the state given in the input register. One may also define a quantum perfect-indistinguishability obfuscator, where the obfuscated states $\mathcal O(C_1)$ and $\mathcal O(C_2)$ are identical. We remark that the condition of equal length can be relaxed to any fixed polynomial (e.g., $|C_1|$ can be of length at most $|C_2|^2$.)

The notion of indistinguishability obfuscation may not seem intuitive at first. To begin to see why it is a useful definition, we now show that it is equivalent to a semantic definition of obfuscation. To achieve this, we first need to properly define functional equivalence for state ensembles. If we fix a quantum translator $(\mathcal O, \mathcal J)$, then the QPT $\mathcal J$ defines a way to implement functionality via states. We then say that two state ensembles $\mathcal R_X$ and $\mathcal S_X$ are functionally equivalent if
$$
\bigl\| \algo J_{\rho_x} - \algo J_{\sigma x} \bigr\|_\diamond \leq \negl(|x|)
$$
for all but finitely many $x \in X$, $\rho_x \in \mathcal R_X$ and $\sigma_x \in \mathcal S_X$. In this case, we will write $\mathcal R_X \cong_{\algo J} \mathcal S_X$. This allows us to compare the relative strengths of one obfuscated ensemble versus another, as follows.

\begin{definition}\label{def:best-possible}
A \textbf{quantum statistical (resp., computational) best-possible obfuscator} is a quantum translator $(\mathcal O, \mathcal J)$ such that for any QPT $\algo A$ there exists a QPT $\algo S$ with the following property: for every circuit ensemble $\mathcal C_X$ and any uniform state ensemble $\mathcal R_X$ which is functionally-equivalent to $\mathcal O(\mathcal C_X)$ and has same-size states\footnote{meaning that, for each $x \in X$, the corresponding states in the two ensembles have the same number of qubits.}, we have that $\algo A(\mathcal O(\mathcal C_X))$ and $\algo S(\mathcal R_X)$ are statistically (resp., computationally) indistinguishable.
\end{definition}

This definition captures the relative ``leakage'' of the obfuscated state ensemble: among all functionally-equivalent state ensembles (i.e., potential obfuscations), the best-possible obfuscation is the ensemble that leaks the least. From this point of view, we think of $\algo A$ as a ``learner'' which tries to learn something from the obfuscated ensemble, and $\algo S$ as a ``simulator'' which can learn the same thing as $\algo A$, but from any other functionally-equivalent ensemble.

\begin{proposition} Let $(\algo O, \algo J)$ be a polynomial-time quantum translator. Then $(\algo O, \algo J)$ is a quantum statistical (resp., computational) best-possible obfuscator if and only if it is a quantum statistical (resp., computational) indistinguishability obfuscator.
\end{proposition}
\begin{proof}
First, let $(\algo O, \algo J)$ be a best-possible obfuscator. Set $\mathcal A$ to be the trivial learner which simply implements the identity operator, and let $\mathcal S$  be the corresponding simulator. Let $\mathcal C$ and $\mathcal D$ be two functionally-equivalent circuit ensembles with same-size circuits. Note that their obfuscations are functionally-equivalent (i.e., $\algo O(\mathcal C) \cong_{\algo J} \algo O(\mathcal D)$) same-size ensembles. By two applications of the best-possible property (one on the left, and one on the right), we get
$$
\algo O(C) = \algo A(\algo O(C)) \approx \algo S(\algo O(D)) \approx \algo A(\algo O(D)) = \algo O(D)\,,
$$
where $\approx$ denotes the appropriate form of indistinguishability (statistical or computational.) It follows that $(\algo O, \algo J)$  is an indistinguishability obfuscator.

For the other direction, let $(\algo O, \algo J)$ be an indistinguishability obfuscator. Given a learner $\algo A$, define a simulator $\algo S$ as follows. Let $(\mathcal C, \mathcal R)$ be a (circuit, ensemble) pair as in the best-possible definition. Since $\mathcal R$ is a uniform state ensemble, there is a corresponding circuit ensemble $\mathcal D$ for preparing it. Given a circuit $D \in \mathcal D$ and the corresponding circuit $J \in \algo J$ of the interpreter, we can then build a circuit $D \circ J$ which is functionally equivalent to the corresponding circuit $C \in \mathcal C$ (by the definition of a quantum translator.) Let $\mathcal {D \circ J}$ be the corresponding circuit family, and define $\algo S ( \mathcal R ) = \algo S ( \algo O ( \mathcal {D \circ J}))$. Now, by the indistinguishability property, $\algo O(\mathcal C) \approx \algo O (\mathcal {D \circ J})$, from which it follows that 
$$
\algo A( \algo O(\mathcal C)) \approx \algo A( \algo O(\mathcal {D \circ J})) = \algo S(\mathcal R)\,,
$$
as desired.
\end{proof}

We remark that the forward implication did not require the obfuscator to be efficient.

\subsection{Impossibility of statistical obfuscation}

In this section, we show that efficient perfect or statistical indistinguishability obfuscation is impossible. We begin by recalling the following computational problems about distinguishability of circuit ensembles and state ensembles, and some relevant complexity-theoretic results.

\begin{problem} \textsf{\emph{Circuit distinguishability}}.\\
\indent Input: two quantum circuits $C$ and $D$; parameter $\epsilon > 0$. \\
\indent Promise: $\|C - D\|_\diamond$ is greater than $2 - \epsilon$ or less than $\epsilon$.\\
\indent Output: YES in the former case and NO in the latter.
\end{problem}

\begin{problem} \textsf{\emph{Quantum State Distinguishability}}.\\
\indent Input: $m$-qubit quantum circuits $C_0$ and $C_1$, positive integer $k \leq m$ and parameters $a,b$ with $a<b^2$.\\
\indent Promise: let $\rho_i = \tr_{(k+1, m)}[C_i\ket{0^m}\bra{0^m}C_i^\dagger]$; then $\|\rho_0 - \rho_1\|_\emph{tr}$ is greater than $b$ or less than $a$.\\
\indent Output: YES in the former case and NO in the latter.
\end{problem}

\begin{theorem}\label{thm:QCD}
\emph{\cite{RW05}} The problem \textsf{\emph{Circuit distinguishability}} is QIP-complete for every $\epsilon > 0$. 
\end{theorem}

\begin{theorem}\label{thm:QSD}
\emph{\cite{Watrous02}} The problem \textsf{\emph{Quantum State Distinguishability}} is QSZK-complete.
\end{theorem}

The following theorem is a straightforward matter of assembling the above results together with the definition of indistinguishability obfuscation.

\begin{theorem}
If there exists a polynomial-time quantum statistical indistinguishability obfuscator, then PSPACE is contained in QSZK.
\end{theorem}
\begin{proof}
We will actually show QIP $\subset$ BQP$^{\text{QSZK}}$; since BQP is contained in QSZK and QIP = PSPACE~\cite{JJUW10}, the result will follow. We will solve \textsf{Circuit distinguishability} using a subroutine that solves \textsf{Quantum State Distinguishability}. Given the classical descriptions of $C$ and $D$, we run the obfuscator $\algo O$ on both to get states $\algo O(C)$ and $\algo O(D)$. Now we apply the subroutine, and output its result.

Note that, if $C$ and $D$ are elements of a pair of functionally-equivalent circuit ensembles, then $\algo O(C)$ and $\algo O(D)$ will be indistinguishable and the subroutine will output YES. On the other hand, if $C$ and $D$ are far from being functionally equivalent, then there exists some input on which they differ significantly. It follows that $\algo O(C)$ and $\algo O(D)$ must be distinguishable; if they were not, then $\algo J_{\algo O(C)}$ and $\algo J_{\algo O(D)}$ would be functionally equivalent, contradicting the translator conditions in the definition of the obfuscator.
\end{proof}

In addition, we can prove an impossibility result for the case of statistical obfuscators which can only obfuscate unitary circuits. We first recall the \textsf{Identity Check} problem. Given an $m$-qubit state $\rho$ and indices $l, k \geq 0$, we let $\tr_{l,k} [\rho]$ denote the result of tracing out qubits $l$ through $k$ of $\rho$. We adopt the convention that nothing is traced out (i.e., $\tr_{l,k}[\rho] = \rho$) if $l > m$.

\begin{problem} \emph{\textsf{Identity Check}}.\\
\indent Input: an $n$-qubit unitary quantum circuit $C$ and parameters $a$,$b$ so that $b-a\geq 1/poly(n)$. \\
\indent Promise: $\min_\alpha \| C - e^{i \alpha} I \|$ is less than $a$ or greater than $b$.\\
\indent Output: YES in the former case and NO in the latter.
\end{problem}

\begin{theorem}\label{thm:ID}
\emph{\cite{JWB03}} The problem Identity Check is coQMA-complete. 
\end{theorem}

\begin{theorem}
If there exists a polynomial-time quantum statistical indistinguishability obfuscator for unitary circuits, then coQMA is contained in QSZK.
\end{theorem}
\begin{proof}
We will actually show coQMA $\subset$ BQP$^{\text{QSZK}}$; since BQP is contained in QSZK, the result will follow. Let $a$ and $b$ satisfy $b-a = 1 / \text{poly}(n)$. We will solve Identity Check using a subroutine that solves Quantum State Distinguishability. 

Let $C$ be the input, i.e., a classical description of an $n$-qubit quantum circuit. Create an identity circuit $D$ with an equal number of inputs as $C$, and of equal length to $C$. Let $O_C$ be a circuit that initializes a register with the classical state $\ket{C}$ containing the classical description of $C$, and applies the circuit of $\mathcal O$ which corresponds to the input length $|C|$. Likewise, let $O_D$ be a circuit that initializes a register with the classical state $\ket{D}$ containing the classical description of $D$, and applies the circuit of $\mathcal O$ which corresponds to the input length $|D| = |C|$. Note that, after tracing out ancillas, the outputs of these circuits are given by
$$
\tr_\text{anc.} \bigl[O_C\ket{0}\bra{0}O_C^\dagger\bigr] = \mathcal O(C)
\qquad \text{and} \qquad
\tr_\text{anc.} \bigl[O_D\ket{0}\bra{0}O_D^\dagger\bigr] = \mathcal O(D)\,.
$$
Now apply the subroutine for solving quantum state distinguishability to the pair $(O_C, O_D)$. If it says ``close'', we output YES; otherwise we output NO. Let's show that this has solved $(a, b)$-identity-check. Note that the states $\mathcal O(C)$ and $\mathcal O(D)$ must have the same number of qubits, and denote that number by $m$.
\begin{itemize}
\item \textbf{completness.} In this case, the obfuscated states satisfy $\|\mathcal O(C) - \mathcal O(D)\|_\text{tr} \leq \alpha$.  By the definition of the induced trace norm, this implies that $\|\mathcal J_{\mathcal O(C)}^n - \mathcal J_{\mathcal O(D)}^n\|_\diamond \leq \alpha$. By functional equivalence for $C$ and $D$ and the triangle inequality, it follows that $\|U_C - U_D\|_\diamond = \|U_C - I\|_\diamond \leq \alpha$, as desired.

\item \textbf{soundness.} In this case, the obfuscated states satisfy $\|\mathcal O(C) - \mathcal O(D)\|_\text{tr} \geq \beta$. We claim that this implies $\|U_C - U_D\|_\diamond > b$. Suppose this is not the case, i.e., that these operators are in fact close; then by the indistinguishability property, it would follow that $\mathcal O(C)$ and $\mathcal O(D)$ are close as well, a contradiction.
\end{itemize}
The above amounts to a BQP$^\text{QSZK}$ protocol for a coQMA-hard problem, thus placing coQMA in QSZK.
\end{proof}

\subsection{Application: quantum witness encryption}

We now give an interesting application of the surviving case of quantum indistinguishability obfuscation, i.e., the computational variant.

The classical idea of witness encryption was first studied in~\cite{GGSW13}; its connection to indistinguishability obfuscation was first considered in~\cite{GGHRSW13}. In the quantum case, we set up the problem as follows. Suppose Alice wishes to encrypt a quantum state $\rho$, but not to a particular key or for a particular person; instead, the encryption is tied to a challenge question, and anyone that can answer the question correctly can decrypt the plaintext. The question will be of a particularly quantum nature: a correct answer will be a quantum state, e.g., the ground state of some Hamiltonian.

More formally, we consider a QMA language $L$, 
and would like to enable Alice to encrypt her state $\rho$ using a particular problem instance $x$. 
If $x$ is a ``yes" instance, 
then we'd like to allow Bob, who holds a witness state, to be able to decrypt Alice's message.  On the other hand, if $x$ is a no instance, then we demand that no QPT can distinguish between encryptions of any two quantum states  on the same number of qubits.  Interestingly, the definition says nothing about the case where $x$ is a yes instance but a witness is not known. While this may seem counterintuitive, the classical primitive has a number of natural applications (e.g., public-key encryption and identity encryption, see \cite{GGSW13}). It is likely that the quantum primitive has the same or similar applications, but we will not explore that question here.

We now show that witness encryption for QMA is possible, assuming the existence of a quantum computational-indistinguishability obfuscator. It is well-known (see, e.g.,~\cite{GN13}) that QMA contains languages $L$ which have a poly-time uniform circuit family $\mathcal C_L$ with completeness $1-2^{-\Omega(n)}$ and soundness $2^{-\Omega(n)}$. Given an instance $x$ of such a language $L$, and a quantum state $\rho$, consider the quantum circuit $Q_{x,\rho}(\sigma)$, which runs the appropriate circuit from $\mathcal C_L$ and outputs $\rho$ on accept, and $\ket{0^n}$ otherwise.  We claim the computational-indistinguishability obfuscation $\algo{O}(Q_{x,\rho})$ is a valid witness encryption for $L$ and $\rho$.  Correctness of decryption is clear, since the ability to provide a valid witness for $L$ allows Bob to use $\algo{O}(Q_{x,\rho})$ to obtain Alice's state $\rho$.  Of course, if $x$ is not in $L$, then no witness will suffice; more precisely, $\|Q_{x,\rho_1} - Q_{x, \rho_2}\| \leq 2^{-\Omega(n)}$ for any two quantum states $\rho_1, \rho_2$ on the same number of qubits. By the indistinguishability condition of \expref{Definition}{def:indistinguishability}, the obfuscation of $Q_{x,\rho_1}$ will be computationally indistinguishable from the obfuscation of $Q_{x,\rho_2}$.

\section{Conclusion}

\subsection{Open problems}

The central remaining open problem is whether quantum-mechanical means of obfuscation are possible, within the restricted framework placed by our results. In the classical setting, it is known that black-box obfuscation and statistical obfuscation are impossible; on the other hand, it is generally believed that computational-indistinguishability obfuscation can be achieved. In the quantum setting, on the other hand, much less is known. Although our results place some significant restrictions, several possibilities remain. For example, it is conceivable that \emph{information-theoretically secure black-box obfuscation} of a program is possible using quantum-mechanical means. Moreover, such a construction may even enable the user to use the obfuscated state to evaluate the program a polynomial number of times. Such an obfuscator would be tremendously powerful, and would be another example of quantum supremacy in the world of cryptography. It is difficult to imagine how such an obfuscation would function for arbitrary quantum programs (due to restrictions placed by the monogamy of entanglement); it is perhaps more reasonable to imagine it for classical functions. In that case, one could imagine using the state to perform the computation, then copying the (classical) output of the program, and restoring the obfuscated state by undoing the unitary circuit.

In the indistinguishability setting, constructions for obfuscating quantum programs are not known. It is conceivable that classical ideas (such as those in~\cite{GGHRSW13}) can be translated to the quantum setting. However, a key ingredient in~\cite{GGHRSW13} is fully-homomorphic encryption, which is yet to be achieved quantumly---although significant progress in this direction was recently made by Broadbent and Jeffery~\cite{BJ15}.

In terms of applications, our results raise another interesting question: what applications can we envision for obfuscators whose outputs are not cloneable or reproducible? Many of the standard applications are restricted in the case. On the other hand, one may be able to design applications where the inability to copy the obfuscator's output is a desired feature. In fact, one can imagine quantum obfuscators where the number of executions is highly limited in a manner analogous to one-time programs (see, e.g.,~\cite{BDG13}); this could also be of use in cryptographic settings.

\subsection{Acknowledgements}

G.\,A. would like to thank Stacey Jeffery and Anne Broadbent for many engaging discussions. G.\,A. was supported by a Sapere Aude grant of the Danish Council for Independent Research, the ERC Starting Grant "QMULT" and the CHIST-ERA project "CQC". B.\,F would like to thank Brad Lackey, Daniel Apon and Jonathan Katz for helpful conversations.  This work was supported by the Department of Defense.  
\bibliography{QuantumCrypto}

\end{document}